%% file: main.tex
\documentclass[11pt]{article}
\usepackage[utf8]{inputenc}
\usepackage[T1,T2A]{fontenc}
\pdfoutput=1

\usepackage{imakeidx}
\newcommand{\df}[1]{{\bf{#1}}{\index{#1}}}
\newcommand{\mdf}[1]{{#1}{\index{$#1$}}} 

\usepackage{mathrsfs,xspace}

\input{packages.tex}
\usepackage[bottom]{footmisc}
\usetikzlibrary{shapes.geometric}
\usepackage{tikz-3dplot}

\usepackage[margin=1in]{geometry}

\input{adam.tex}

\definecolor{red(colorwheel)}{rgb}{1.0, 0.0, 0.0}

\definecolor{azure(colorwheel)}{rgb}{0.0, 0.5, 1.0}
\newcommand{\igor}{\color{azure(colorwheel)}??}
\definecolor{orange(colorwheel)}{rgb}{1.0, 0.5, 0.0}
\newcommand{\adam}{\color{orange(colorwheel)}}

\newcommand{\bill}{\color{red(colorwheel)} ??}

\newcommand{\csym}{\stackrel{\mathrm{cyc}}{\thicksim}}
\newcommand{\cycsos}{\widetilde\SOS}
\newcommand{\trvna}{tracial von Neumann algebra\xspace}
\newcommand{\NullSS}{Nullstellensatz\xspace}
\newcommand{\NullSSs}{Nullstellens\"atze\xspace}
\newcommand{\relaxgp}{relaxed game polynomial\xspace}
\newcommand{\relaxgps}{relaxed game polynomials\xspace}

\newcommand{\Ugames}{Torically determined games\xspace}
\newcommand{\ugame}{torically determined game\xspace}
\newcommand{\ugames}{torically determined games\xspace}

\newcommand{\tgame}{torically determined game\xspace}
\newcommand{\tgames}{torically determined games\xspace}
\newcommand{\tdby}{torically determined by\xspace}
\newcommand{\TGames}{Torically Determined Games\xspace}
\newcommand{\dby}{determined by\xspace}
\newcommand{\detset}{determining set\xspace}
\newcommand{\detsets}{determining sets\xspace}

\newcommand{\cB}{\mathcal{B}}

\newcommand{\cC}{\mathcal{C}}
\newcommand{\cX}{\mathcal{X}}
\newcommand{\cE}{\mathcal{E}}

\def\cZ{{\mathcal Z}}

\usepackage{caption}
\usepackage{appendix}

\def\cH{{\mathcal H}}

\def\bs{\bigskip}

\def\CC{{\mathbb C}}
\def\C{\mathbb C}
\def\R{\mathbb R}

\newcommand{\strat}{S}
\newcommand{\costrat}{\mathcal{S}_{\mathrm{co}}}
\newcommand{\game}{\mathcal{G}}

\newcommand{\gr}{\mathcal{Y}} 
\newcommand{\br}{\mathcal{N}} 
\DeclareMathOperator{\synch}{synch}
\newcommand{\sbr}{\synch\mathcal{B}} 
\newcommand{\uGA}{\mathscr{U}}
\def\uGI{{\mathscr{I}}}
\newcommand{\siga}[3]{x(#3)^{#1}_{#2}} 
\newcommand{\proja}[3]{e(#3)^{#1}_{#2}}
\newcommand{\projr}[3]{E(#3)^{#1}_{#2}}
\newcommand{\projrset}{\cE}
\newcommand{\cycla}[2]{c({#2})_{#1}}
\newcommand{\lcycla}[3]{\cycla{#1}{#2}^{(#3)}}
\newcommand{\xora}[2]{x^{(#2)}_{#1}}
\newcommand{\score}{V}
\newcommand{\covalue}{\omega^*_{\mathrm{co}}}
\newcommand{\questions}{\mathcal{Q}}

\newcommand{\clause}{h}
\newcommand{\clausegp}{H}
\newcommand{\validpoly}{\mathcal{Y}}
\newcommand{\invalidpoly}{\mathcal{N}}
\newcommand{\LIg}[1]{\cLI\left(#1\right)}
\newcommand{\cRI}{\mathfrak{R}}
\newcommand{\rep}{\pi}
\newcommand{\algebra}{\mathscr{A}}

\newcommand{\ssalg}{\mathcal{C}} 

\newcommand{\polyn}{F}
\newcommand{\pol}{f}
\newcommand{\monom}{r}
\newcommand{\monomgp}{\Delta} 
\newcommand{\state}{\psi}

\newcommand{\gagp}{G} 

\newcommand{\dset}{\mathcal{F}} 
\newcommand{\delt}{f} 

\newcommand{\clauseset}{\mathscr{H}}

\newcommand{\rgp}{{\tilde{\Phi}}}
\newcommand{\gp}{\Phi}
\newcommand{\betagp}{B}

\DeclareMathOperator{\Alg}{Alg}

\DeclareMathOperator{\supp}{supp}
\DeclareMathOperator{\LT}{LT}
\DeclareMathOperator{\LC}{LC}
\DeclareMathOperator{\SOS}{SOS} 
\DeclareMathOperator{\clos}{cl}
\DeclareMathOperator{\mon}{Mon}
\DeclareMathOperator{\cone}{cone}

\def\ben{\begin{enumerate} }
\def\een{\end{enumerate} }

\def\sec{\section}
\def\ssec{\subsection}
\def\sssec{\subsubsection}

\title{Noncommutative \NullSSs and Perfect Games}
\author{Adam Bene Watts\thanks{University of Waterloo. \texttt{adam.benewatts1@uwaterloo.ca}} 
\and J. William Helton\thanks{University of California San Diego. \texttt{helton@math.ucsd.edu}}
\and Igor Klep\thanks{University of Ljubljana, Slovenia \texttt{igor.klep@fmf.uni-lj.si}}}

\makeindex

\begin{document}

\maketitle


\begin{abstract}
The foundations  of classical Algebraic Geometry and Real Algebraic Geometry
are the \NullSS and Positivstellensatz.
Over the last two decades the basic  analogous theorems
for matrix and operator theory (noncommutative  variables) have emerged.
This paper concerns 
commuting operator strategies for nonlocal games, recalls NC \NullSS which are helpful, extends these, and applies them to a  very broad collection of games. In the process it brings together results spread over  different literatures, hence rather than being  terse, our style is fairly expository.

The main results of this paper are two characterizations, based on \NullSS, which apply to games with perfect commuting operator strategies. The first applies to all games and reduces the question of whether or not a game has a perfect commuting operator strategy to a question involving left ideals and sums of squares.
Previously, Paulsen and others
translated the study of perfect 
synchronous games to   problems entirely involving a  $*$-algebra.
The characterization we present is analogous, but works for all games.
The second characterization is based on a new \NullSS we derive in this paper. It applies to a class of games we call \tgames, special cases of which are XOR and linear system games. For these games we show the question of whether or not a game has a perfect commuting operator strategy reduces to instances of the subgroup membership problem and, for linear systems games, we further show this subgroup membership characterization is equivalent to the standard characterization of perfect commuting operator strategies in terms of solution groups. Both the general and \tgames characterizations are amenable to computer algebra techniques, which we also develop.

For context we mention that
Positivstellens\"atze 
are behind the
standard 
NPA
upper bound on the  score players can achieve for a  game using a commuting operator strategy.
%
%
This paper
develops  analogous
 NC real algebraic geometry which bears on perfect games.

\end{abstract}

\newpage

\def\bbk{{\mathbb k}}
\def\fae{{\C\langle x\rangle}}

\newcommand{\fa}{\mathbb{C}\langle x\rangle}
\newcommand{\fay}{\mathbb{C}\langle x,\xi\rangle}

\def\faxyz{\mathbb C\langle x,y,z\rangle}
\def\faxyzpsi{\mathbb C\langle x,y,z,\xi\rangle}

\newcommand{\ide}{\mathfrak I}
\newcommand{\lide}{\mathfrak L}

\def\cLI{{\lide}}
\def\bbE{{\mathbb E}}
\def\cexp{{\mathbb E}}

\sec{Introduction}


A nonlocal game is a test performed between a verifier and $k$ players, in which the verifier tests the players' ability to produce correlations without communicating. In a round of the game the verifier sends questions to the players and the players return responses to the verifier. The list of questions and responses is then scored according to a function known by both the verifier and the players before the game began. By convention, the score achieved lies in the interval $[0,1]$. The players cooperate to try and achieve the highest possible score, with the challenge that the players can't communicate while the game is in progress and so don't know the questions sent to other players.\looseness=-1

The optimal score the players can achieve on a nonlocal game $\game$ depends on the resources the players share. If the players share only classical randomness the optimal score they can achieve in expectation is called the the classical value of the game, denoted $\omega(\game)$. If players share an arbitrary state in a (possibly infinite dimensional) Hilbert space and can make commuting measurements on it the optimal score they can achieve is called the commuting operator value of the game, denoted $\covalue(\game)$. The supremum value achievable by players who make commuting measurements on a state in a finite dimensional entangled space is called the quantum value, denoted $\omega^*_{\textrm{q}}(\game)$.
These three values can all differ, though the inequalities $\omega(\game) \leq \omega^*_{\textrm{q}}(\game) \leq \covalue(\game)$ are always satisfied. 




Starting roughly in this century the
 classical subject of real algebraic geometry has been extended to matrix and operator (noncommutative) variables.
Here  inequalities and equalities
are explained by being equivalent to  algebraic formulas,
often involving Sums of Squares (\df{SOS}). 
These go under the names of 
\df{Positivstellensatz}
for inequalities 
and 
\df{\NullSS} for equations.
Of course finding 
quantum strategies for games 
leads to many such noncommutative
(\df{NC}) inequalities and equalities.\looseness=-1

In  this paper we describe how the 
well developed NC real algebraic geometry theory applies and integrates with 
nonlocal games and  commuting operator strategies for them. We show a connection between NC \NullSS and whether or not a nonlocal game has a perfect commuting operator solution (i.e., $\covalue(\game) = 1$).
  This connection gives a new algebraic characterization which applies to all nonlocal games with commuting operator value exactly equal to one. This characterization provides a unified algebraic framework through which several previous results concerning the commuting operator value of nonlocal games can be understood. For a large class of games it also reduces the question of whether or not a game has perfect commuting operator value to an instance of the subgroup membership problem, providing a potential starting point for the investigation of several yet-to-be studied families of games.\looseness=-1

For context, 
Positivstellens\"atze have long played a major role in the study of nonlocal games in that they are behind the
standard \cite{navascues2008convergent,doherty2008quantum}
upper bound on the commuting operator value of a game.
Underlying this bound is one of the earliest NC Positivstellens\"atze, 
\cite{HM04}.
This paper
turns its  attention to
developing the analogous
 NC real algebraic geometry which bears on perfect games.


In the remainder of this introduction we introduce some new terminology, review some previous results concerning the commuting operator value of nonlocal games, and then give formal statements of some of our main results. 

\paragraph{Algebraic Description of the Commuting Operator Value}

A commuting operator strategy for a nonlocal game is a description of how players can use commuting operator measurements  to map questions sent by the verifier to responses. Formally, a (commuting operator) strategy can be specified by a Hilbert space $\cH$, a state $\state \in \cH$ which is shared by the players and projectors $\{ \projr{i}{a}{\alpha} \}$ acting on $\cH$, where $\alpha$ ranges over all players, $i$ ranges over all questions, and $a$ ranges over all responses. The projector $\projr{i}{a}{\alpha}$ can be read as ``the projector corresponding to player $\alpha$ giving response $a$ to question $i$''.

Because the Hilbert space $\cH$ on which they act is arbitrary, it is difficult to reason about the $\projr{i}{a}{\alpha}$ directly. Instead we introduce the universal game algebra $\uGA$, a $*$-algebra generated by variables $\proja{i}{a}{\alpha}$ which satisfy the same relations as the projectors $\projr{i}{a}{\alpha}$, for example, that $\proja{i}{a}{\alpha}$ and $\proja{j}{b}{\beta}$ commute for any $\alpha \neq \beta$.\footnote{The $*$-algebra $\uGA$ is isomorphic to a group algebra and has appeared before in other contexts. For example, in \cite{lupini2020perfect} an algebra closely related to $\uGA$ was denoted $\mathcal{A}(X,A)$.}$^,$\footnote{A full set of relations for $\uGA$ is listed in \Cref{sssec:uGA}.}  
Commuting operator strategies can then be specified by tuples $(\rep,\state)$, consisting of a $*$-representation $\rep$ mapping $\uGA$ to bounded operators on a Hilbert space $\cH$, along with a state $\state \in \cH$. When specified in this way, it is understood that projectors $\projr{i}{a}{\alpha}$ are given by $\rep(\proja{i}{a}{\alpha})$ and that $\state$ gives the state shared by the players.

\def\ax{\langle x\rangle} 

\paragraph{Other Characterizations of Perfect Commuting-Operator Strategies}
Several other papers have considered the problem of deciding whether or not a game has a perfect commuting operator strategy and given criteria which determine the existence of perfect commuting operator strategies for specific families of nonlocal games. We review some of those families of games and the associated characterizations below. 
\begin{itemize}

\item Linear systems games are two player games based around system of $m$ linear equations on $n$ variables. 
In~\cite{cleve2017perfect} it was shown that deciding existence of a perfect commuting operator strategy for a binary linear systems game was equivalent to solving an instance of the word problem on a group called the solution group of the game. 

\item XOR games are $k$ player games which, similarly to linear system games, test satisfiability of a system of $m$ binary equations on $kn$ variables. 

In~\cite{watts20203xor} it was shown that deciding the existence of a perfect commuting operator strategy for an XOR game was equivalent to solving an instance of the subgroup membership problem on a group called the game group.



\item Synchronous games are two player nonlocal games which include ``consistency-checks'', where Alice and Bob are sent the same question and win iff they send the same response. Other than these consistency checks, the questions and winning responses involved in a synchronous game are arbitrary.   In~\cite{paulsen2016estimating} it was shown that there was a perfect commuting operator strategy for a coloring game iff a $*$-algebra associated with a single player's operators could be represented into a $C^*$-algebra with a faithful trace. In \cite{helton2017algebras} and \cite{kim2018synchronous} this was generalized to synchronous games.

\end{itemize}


\paragraph{Our Results}

The main results of this paper are two theorems giving algebraic characterizations of games with perfect commuting operator strategies.  

A key concept introduced on the way to proving these theorems is the notion of a game being \dby a set of elements $\dset\subseteq\uGA$. Formally we say a game $\game$ is \dby a set of elements $\dset$ if, for any commuting operator strategy $(\rep,\state)$, we have that $(\rep,\state)$ is a perfect commuting operator strategy for $\game$ iff 
$
    \rep(\delt) \state = 0
$
for all $\delt \in \dset$. We also note that any game $\game$ is naturally determined by two sets of elements. The first, $\invalidpoly$, consists of elements corresponding to projectors onto responses which obtain a score less than $1$ on questions asked by the verifier, while the second, $\validpoly$, consists of elements $y-1$ with each element $y$ corresponding to projectors onto responses which obtain a score of exactly $1$. Details of this construction are given in \Cref{sssec:valid_invaid_response_det_sets}.

Our first major theorem follows from combining the notion of sets of elements which determine a game with a result in noncommutative algebraic geometry known as a \NullSS. To state the result formally, 
let $\LIg{\cX}$ denote the left ideal of $\uGA$ generated by $\cX$ for any set of elements $\cX \subseteq \uGA$ and $\SOS_\uGA$ denote sums of squares in the algebra $\uGA$. Then the following result holds:

\begin{thm} \label{thm:GamesNullSSRestated}
For a nonlocal game $\game$ \dby a set $\dset \subseteq \uGA$ the following are equivalent:
\begin{enumerate}[\rm(i)]
\item
$\game$ has a perfect commuting operator strategy;
\item $-1 \notin \LIg{\dset} + \LIg{\dset}^* + \SOS_\uGA $.
\end{enumerate}
\end{thm}
\Cref{thm:GamesNullSSRestated} combined with the natural determining sets $\invalidpoly$ and $\validpoly$, gives a fully algebraic characterization of nonlocal games with perfect commuting operator strategies. This characterization is analogous to the characterization of synchronous games given in \cite{helton2017algebras}, but works for all games. For the special case of synchronous games we show in \Cref{sec:Synchronous} that the characterizations of \Cref{thm:GamesNullSSRestated} and \cite{helton2017algebras} are equivalent.

The second major theorem focuses on a general class of games on which \Cref{thm:GamesNullSSRestated} can be simplified further. A game $\game$ is called a \tgame
if there exists a group $G$ with $\uGA \cong \mathbb{C}[G]$ and $\game$ is determined by a set of elements 
\begin{align}
    \dset = \{\beta_i g_i - 1\}
\end{align}
with each $\beta_i \in \mathbb{C}$ and $g_i \in G$.
We call the elements $\beta_i g_i$ clauses of $\dset$, and let $\clauseset = \{\beta_i g_i\}$ be the set of all the clauses of $\dset$. We give the following characterization of \tgames with perfect commuting operator strategies:  
\begin{thm} \label{thm:perfect_unitary_games_Restated}
Let $\game$ be a game \tdby a set of elements $\dset$ with clauses $\clauseset$. Then $\game$ has a perfect commuting operator strategy iff the following equivalent criteria are satisfied:
\begin{enumerate}[\rm(i)]
    \item $-1 \notin  
    \LIg{\dset}+\LIg{\dset}^*$;
    \item The subgroup $H$ of $\uGA$ generated by $\clauseset \cup \clauseset^*$
    meets $\CC$ only in $1$.
\end{enumerate}
\end{thm}
Condition (i) makes it clear \Cref{thm:perfect_unitary_games_Restated} can be viewed as a version of \Cref{thm:GamesNullSSRestated} for \tgames which holds without the SOS term. Additionally, condition (ii) reduces the characterization of perfect commuting operator strategies in terms of $*$-algebras given in \Cref{thm:GamesNullSSRestated} to one entirerly in terms of groups. In the paper we show that both linear systems games and XOR games are \tgames, and that \Cref{thm:perfect_unitary_games_Restated} recovers the algebraic characterizations of these games given in \cite{cleve2017perfect,watts20203xor} respectively. We also show that \Cref{thm:perfect_unitary_games_Restated} lets us extend the algebraic characterization of both XOR and binary linear systems games to more general games based on linear equations Mod $r$ for any integer $r$.\footnote{In \cite{cleve2017perfect} it was already observed that the characterization of binary linear systems games presented could be generalized to any system of equations mod $p$ for any prime $p$. This generalization is given explicitly in \cite{goldberg2021synchronous}. }

Both \Cref{thm:GamesNullSSRestated,thm:perfect_unitary_games_Restated} allow new algorithms for identifying nonlocal games with perfect commuting operator strategies. In this paper we discuss one such algorithm, based on Gr\"obner bases, and give some sample applications. We note that, unlike the upper bounds coming from the ncSoS hierarchy, these Gr\"obner bases algorithms can both prove a game has commuting operator value strictly less than $1$ and identify some games with commuting operator value exactly equal to $1$.\footnote{The question of whether a game has perfect commuting operator value is undecidable~\cite{slofstra2020tsirelson}, meaning these algorithms (or any algorithms!) cannot always identify games with commuting operator value one, but there are many examples where they do.}

\paragraph{Recent Related Work} In \cite{lupini2020perfect} the characterization of perfect commuting operator strategies for synchronous games in terms of representations of a $*$-algebra into a $C^*$ algebra with a faithful trace was generalized to a class of games known as imitation games, and a purely algebraic characterization of perfect commuting operator strategies was giving for a subclass of imitation games known as mirror games. In this paper we do not try and reconcile these characterizations with \Cref{thm:GamesNullSSRestated,thm:perfect_unitary_games_Restated}, but we do note it as an interesting problem worth future study. In \cite{goldberg2021synchronous} the solution group characterization of perfect commuting operator strategies for linear systems game given in~\cite{cleve2017perfect} was related to the $C^*$-algebra characterization of perfect commuting operator strategies for a synchronous variant of linear systems games presented in~\cite{kim2018synchronous}. 
The relationship between these algebraic objects is similar to the relationships presented in \Cref{sec:Synchronous,sec:LinearSysGames} of this paper.

\paragraph{Acknowledgements} The authors would like to thank Vern Paulsen and William Slofstra for helpful discussions. IK was supported by the Slovenian Research Agency grants J1-2453, N1-0217 and P1-0222. ABW was supported by NSF grant CCF-1729369.

\paragraph{Paper Outline}
In \Cref{sec:Math_Background} we review some general mathematical definitions necessary to understand the main results of the paper. In \Cref{sec:NG_defn} we introduce terminology related to nonlocal games and define key algebraic objects. In \Cref{sec:PerfectGamesNullSS} we introduce \NullSSs and prove \Cref{thm:GamesNullSSRestated}. In \Cref{sec:NoSOSNullSS} we simplify the \NullSSs from the previous section and prove \Cref{thm:perfect_unitary_games_Restated},  then apply it to XOR and Mod $r$ games. \Cref{sec:LinearSysGames} shows that \Cref{thm:perfect_unitary_games_Restated} is equivalent to the characterization of linear systems games given in \cite{cleve2017perfect}. \Cref{sec:Synchronous} shows that, in the special case of synchronous games, \Cref{thm:GamesNullSSRestated} is equivalent to the characterization of synchronous games given in~\cite{helton2017algebras}. \Cref{sec:GBAlgorithmFull} introduces Gr\"obner basis algorithms for identifying games with perfect commuting operator value based on \Cref{thm:GamesNullSSRestated,thm:perfect_unitary_games_Restated} and \Cref{sec:GBExamples1} gives an example application of our \NullSSs to quantum graph coloring.

There is a table of contents,  an index and a list of notation at the end of the paper.

\section{Math Background}

\label{sec:Math_Background}

\paragraph{Representations}
Given an algebra $\cA$, a \df{representation} of $\cA$ is a 
unital homomorphism $\pi:\cA\to\cB(\cH)$ for some Hilbert space $\cH$.
When $\cA$ is endowed with an involution $*$, then $\pi$ is 
called a $*$-representation if $\pi(a^*)=\pi(a)^*$
for all $a\in\cA$.

\def\ax{\langle x\rangle}

\paragraph{Free algebras}
Let $x=(x_1,\ldots,x_d)$ denote a tuple of noncommuting letters.
Words in $x$ form the \df{free monoid} $\ax$, including the empty word denoted $1$. 
A \df{noncommutative (nc) polynomial} is a linear combination of words;
the algebra of nc polynomials is 
the \df{free algebra} $\C\ax$. We endow it with the involution 
$*$ fixing $\R\cup\{x\}$.

\begin{ex}
Any representation $\pi$ of $\C\ax$ is described by a tuple of operators
$X\in\cB(\cH)^d$ through $\pi(x_j)=X_j$. Then $\pi$ is a $*$-representation if and only if the $X_j$ are self-adjoint.
\end{ex}

\paragraph{Group algebras}
Given an abstract group $G$, the \df{group algebra} $\C[G]$ has $G$ as
a vector space basis, and the multiplication is extended from the one in $G$ by linearity/distributivity. 
Thus elements  $p\in\C[G]$ are finite linear combinations 
\begin{align}
    p= \sum_j s_j g_j,
\end{align}
where $s_j\in\C$ and $g_j \in G$.
The involution on $\C[G]$ is 
complex conjugation on $\C$ and  $g\mapsto g^{-1}$ for $g\in G.$
Hence a $*$-representation of a group algebra
$\pi:\C[G]\to\cB(\cH)$ is a ring homomorphism such that
$\pi(g)$ is a unitary operator for all $g\in G$. 

\begin{ex}
Each group algebra admits a left regular $*$-representation. Namely, let $\ell^2(G)$ be the Hilbert space with Hilbert space basis $G$. 
Then the left regular representation $\lambda$ defined as
\[
\lambda(g) (\sum_j a_j g_j) = \sum_j a_j g^{-1} g_j
\qquad \text{for } \sum_j a_j g_j \in \ell^2(G) .
\]
Thus each $g\in G$ induces a unitary operator
$\lambda(g)\in \cB(\ell^2(G))$
and thus by linearity $\lambda:\C[G]\to\cB(\ell^2(G))$ 
is a $*$-representation.
\end{ex}

\sec{Nonlocal Game Definitions}
\label{sec:NG_defn}
    
This section gives an overview of all the terminology used to discuss nonlocal games in this paper. \Cref{ssec:general_notation} gives a semi-formal introduction to the language used to describe nonlocal games. \Cref{ssec:technical_definitions} gives technical definitions which are key to this paper. In \Cref{ssec:Algebraic_Picture} we introduce an algebraic framework which we will use to describe nonlocal games and their commuting operator strategies. In \Cref{ssec:Characterizing_perfect_games} we describe the condition that a nonlocal game has perfect commuting operator strategy in terms of some of the notation introduced in previous sections.  Finally, in \Cref{ssec:Games_Examples} we describe some well-know families of games using the language introduced in previous sections.

A reader already familiar with nonlocal games can skip \Cref{ssec:general_notation} but should not skip \Cref{ssec:technical_definitions} (or later sections) as conventions are set in those sections which will remain important throughout the paper. 

\ssec{General Notation}
\label{ssec:general_notation}

Nonlocal games describe experiments which test the correlations that can be produced by measurements of quantum systems. A nonlocal game involves a \df{referee} (also called the \df{verifier}) and $k$ \df{players} (also called \df{provers}). In a round of the game, the verifer selects a vector of questions $\vec{i} = (i(1), i(2), \ldots , i(k))$ randomly from a set $\questions$ of possible question vectors, then sends player $\alpha$ question $i(\alpha)$. Each player responds with an answer $a(\alpha)$. The players cannot communicate with each other when choosing their answers. After receiving an answer from each player, the verfier computes a \df{score} $V(\vec{a} \;|\vec{i}) = V(a(1), a(2), \ldots, a(k) | i(1), i(2), \ldots , i(k))$ which depends on the questions selected and answers received. The players know the set of possible questions $S$ and the scoring function $V$. Their goal is to chose a \df{strategy} for responding to each possible question which maximizes their score in expectation. The difficulty for the players lies in the fact that in a given round each player only has partial information about the questions sent to other players.
The set of all questions possible 
is finite; cardinality is denoted  $m$.\looseness=-1

For a given game $G$, the maximum expected score achievable by players is called the \df{value} of the game. The value depends on the resources available to the players. If players are restricted to classical strategies, the value is called the \df{classical value} and denoted $\omega(G)$. If players can make measurements on a shared quantum state (but still can't communicate) the value can be larger and is called the \df{entangled value}. More specifically, if the players shared state lives in a Hilbert space $\cH = \cH_1 \otimes \cH_2 \otimes \cdots \otimes \cH_k$ and the $i$-th player makes a measurement on the $i$-th Hilbert space ($i.e.$ all measurement operators used by the $i$-th player take the form $I _1 \otimes  \cdots \otimes I_{i-1} \otimes X_i \otimes I_{i+1} \otimes \cdots \otimes I_k$ with $X_i$ acting on $\cH_i$) the supremum score the players can obtain is called the \df{quantum value}, denoted $\omega^*_{\textrm{q}}$. In \cite{scholz2008tsirelson} it was shown that this value is equal to the supremum score the players can achieve making measurements on finite dimensional Hilbert spaces.

If the players state and Hilbert space is arbitrary and the only restriction placed on their measurements is all measurement operators for different players commute (enforcing no-communication), the maximum achievable score is called the \df{commuting operator value}, denoted $\covalue$. When the state shared by the players is finite dimensional the tensor product and commuting operator values of a game coincide. In the infinite dimensional case  $\omega^*_{\textrm{tp}} \leq \covalue$, and there exist games for which the inequality is strict~\cite{ji2020mip,fritz,junge}.

A notation summary is: 
\beq 
k:= \# \text{players}, \ \ n:= \# \text{questions}
\ \ m:= \#\text{responses per question} .
\eeq

\ssec{Basic Definitions}
\label{ssec:technical_definitions}

\sssec{Commuting Operator Strategies}

We start with a definition of a commuting operator strategy for nonlocal games. The setting is a separable Hilbert space $\cH$, possibly infinite dimensional; measurements, described by bounded operators in $\cB(\cH)$  and  states, 
positive semidefinite $\rho \in \cB(\cH)$ with trace 1.
Such operators $\rho$ are often called \df{density operators} or (normalized) \df{trace class operators}.
Of course
$\ell_{\rho}(W)= \tr(W \rho)$
defines 
a positive linear functional $\ell_\rho: \cB(\cH) \to \CC$ with $\ell(I)=1$. 
A \df{pure state} 
corresponds to
rank one $\rho$, that is 
$\rho= \psi^* \psi$
with $\psi$ a unit vector in $\cH$.
In this case $\ell_\rho$ is of the form $W\mapsto \psi^* W\psi$.

\bs 

\begin{defn} \label{defn:co_strat}
A \df{commuting operator strategy} $\strat$ for a $k$-player, $n$-question, $m$-response nonlocal game is defined by $(\cH,\rho , \projrset(1), \projrset(2), \ldots , \projrset(k))$ where $\cH$ is a separable Hilbert space, $\rho$ a state, and each $\projrset(\alpha)$ for $\alpha\in
[k]$
is a set of projectors acting on 
$\cH$,
\begin{align}
    \projrset(\alpha) = \{ \projr{i}{a}{\alpha} \mid i \in [n], a \in [m] \}
\end{align}
which satisfy
\begin{subequations}
\begin{align}
    [\projr{i}{a}{\alpha}, \projr{k}{b}{\beta}] & = 0 \quad \forall \; \alpha \neq \beta, \\
    \sum_{a \in [m]} \projr{i}{a}{\alpha} & = 1 \quad \forall \; \alpha \in [k], i \in [n]. \label{eq:proj_sum_to_1} 
\end{align}
\end{subequations}
\end{defn}

Note that as a consequence of \Cref{eq:proj_sum_to_1} we have $\projr{i}{a}{\alpha}\projr{i}{b}{\alpha} = 0$ and hence $[\projr{i}{a}{\alpha},\projr{i}{b}{\alpha}] = 0$ for any $\alpha, i$ and $a \neq b$. \\

\noindent
Conventions:
\begin{center}
\df{$\alpha, \beta$ variables label players}, \ \  \
\df{$i,j$ variables label questions},  and 
\\
\df{$a,b$ variables label responses}.
\end{center}

\sssec{Games and their Commuting Operator Value}

\label{sssec:standard_games_notation}

A $k$-player, $n$-question, $m$-response nonlocal game $\game = (V, \mu)$ is specified by a scoring function 
\begin{align}
    V : [n]^k \times [m]^k \rightarrow [0,1]
\end{align}
and a probability distribution $\mu$ on $[n]^k$. The \df{score} a strategy $\strat$ obtains on a game $\game = (V,\mu)$ is given by 
\begin{align}
    \mdf{\score(\game,\strat)} = \sum_{\vec{i} \in [n]^k}  \sum_{\vec{a} \in [m]^k} \mu \left(\vec{i}\right) V\left(\vec{i}, \vec{a}\right) \Tr[\prod_{\alpha \in [k]} \projr{i(\alpha)}{a(\alpha)}{\alpha} \rho]
\end{align}
The commuting operator value $\covalue(\game)$ of a game is defined to be the supremum value achieved over commuting operator strategies, so 
\begin{align}
    \covalue(\game) = \sup_{\strat \in \costrat} \score(\game, \strat). \label{eq:covalue_physics}
\end{align}
where we have defined $\costrat(k,n,m)$ to be the set of all $k$-player, $n$-question, $m$-response commuting operator strategies. Often $k$, $n$, $m$ are implied from context, and we just write $\costrat$.

\ssec{The Algebraic Picture}
\label{ssec:Algebraic_Picture}

In this paper we think of strategies $\strat \in \costrat$ as arising from from representations of an algebra we call the universal game algebra. We define that algebra next.

\sssec{Universal Game Algebra}
\label{sssec:uGA}
Here we define the \df{Universal Game Algebra}
$\uGA$ in terms of various generators and relations.
These relations define a two sided ideal $\uGI$ which lives inside the $*$-algebra of all noncommutative polynomials in the generators. We call $\uGI$ the
\df{Universal Game Ideal}.
These relations reflect the algebraic properties of projectors or related algebraic objects.

\paragraph{Projection Generators}
Define $\uGA$ to be the $*$-algebra with generators 
$e(\alpha)^i_{a}$
which satisfy relations
\begin{subequations}\label{eq:defUA}
\begin{align}
    [e(\alpha)^i_a, e(\beta)^j_b] &= 0 \quad \forall \;  i,j, a, b, \alpha \neq \beta, \\
    \left(e(\alpha)^i_a\right)^2 &= \left(e(\alpha)^i_a\right)^* = e(\alpha)^i_a, \\
    e(\alpha)^i_a e(\alpha)^i_b &= 0 \quad \forall \; i, a \neq b,  \\
    \sum_a e(\alpha)^i_{a} &= 1 \quad \forall \; \alpha, i,
\end{align}
\end{subequations}
with $i,j \in [n]; a,b \in [m],$ and $\alpha \in [k]$. (Technically we should define a different Universal Algebra for every different value $n,m$ and $k$, so $\uGA = \uGA(n,m,k)$. We frequently omit this detail when $n,m$ and $k$ are clear from context.)

\def\faee{{\C\langle e\rangle}}

Alternately, we can describe $\uGA$ as a quotient of the free $*$-algebra. Namely, let $e$ denote the tuple 
$$
e= (e(\alpha)_a^i ))_{i,a,  \alpha},
$$
and define the universal game ideal $\uGI$ to be the 2 sided ideal in 
the free algebra $\faee$, defined by the 
polynomials 
\begin{align}
   \big\{ [e(\alpha)^i_a, e(\beta)^j_b],  \; 
    \left(e(\alpha)^i_a\right)^2 
    - e(\alpha)^i_a, \ 
    e(\alpha)^i_{a'} e(\alpha)^i_{b'},\ 
    \sum_a e(\alpha)^i_{a} - 1 \big\}
\end{align}
where the indices run through all
$ i,j, a, b, a'\neq b', \alpha \neq \beta .$ 
Then   $\uGA = \faee/ \uGI$; note that $\uGA$ comes equipped with
an involution induced by $\big(e(\alpha)^i_a\big)^*= e(\alpha)^i_a$.

\bs 

There are two common changes of variables we will use when describing the algebra $\uGA$. 

\paragraph{Signature Matrix Generators}
 The first change of variables is to generators satisfying the algebraic properties of signature matrices, defined by
\begin{align}
    \siga{i}{a}{\alpha} := 2\proja{i}{a}{\alpha} - 1.
\end{align}
These variables satisfy relations 
\begin{subequations}
\begin{align}
    \siga{i}{a}{\alpha}\siga{j}{b}{\beta} &= \siga{j}{b}{\beta}\siga{i}{a}{\alpha} \quad \forall \; i,j,a,b, \alpha \neq \beta,  \label{eq:siga_commute} \\
    (2\siga{i}{a}{\alpha}-1)(2\siga{i}{b}{\alpha}-1) &= 0 \quad \forall \; i,a\neq b, \alpha \text{ and }\label{eq:adds} \\
    \sum_{a} \siga{i}{a}{\alpha} &= - (m-2),
    \label{eq:siga_sum} \\
    \left(\siga{i}{a}{\alpha}\right)^* &= \siga{i}{a}{\alpha}, \label{eq:siga_inverse} \\
    \left(\siga{i}{a}{\alpha}\right)^2 &= 1 \quad \forall \; i,a, \alpha. \label{eq:siga_square}
\end{align}
\end{subequations}
It is straightforward to check that the set of relations above gives a defining set of relations for the algebra $\uGA$ written in terms of the $\siga{i}{a}{\alpha}$. 

\paragraph{Cyclic Unitary Generators} The second change of variables is to cyclic unitary generators, defined by 
\begin{align}
    \cycla{j}{\alpha} := \sum_a \exp(\frac{2\pi a i}{ m}) \proja{j}{a}{\alpha}.
\end{align}
(In this subsection $i$ refers to the imaginary unit, and $j$ is used to index questions). These elements satisfy relations 
\begin{subequations}
\begin{align}
    \cycla{j}{\alpha}\cycla{j'}{\beta} &= \cycla{j'}{\beta}\cycla{j}{\alpha} \quad \forall \; j,j', \alpha \neq \beta, \label{eq:cycobsrel1} \\
    \left(\cycla{j}{\alpha}\right)^* &= \left(\cycla{j}{\alpha}\right)^{-1} \label{eq:cycobsrel2}, \\
    \left(\cycla{j}{\alpha}\right)^m &= 1. \label{eq:cycobsrel3}
\end{align}
\end{subequations}
Routine calculation using the inverse transformation 
\begin{align}
    \proja{j}{a}{\alpha} = \frac{1}{m} \sum_b \left(\exp(\frac{-2\pi a i}{m}) \cycla{j}{\alpha}\right)^b
\end{align}
shows that \Cref{eq:cycobsrel1,eq:cycobsrel2,eq:cycobsrel3} also form a defining set of relations for $\uGA$. 

We can also define a group $\gagp$ generated by elements $\cycla{j}{\alpha}$ which satisfy the relations defined in \Cref{eq:cycobsrel1,eq:cycobsrel2,eq:cycobsrel3}. Then we can write that $\uGA = \mathbb{C}[\gagp]$, making it clear that $\uGA$ is a group algebra. 

\paragraph{Two Answer Games}

An notable special case occurs when the answer set of the game contains only $2$ responses. In this case the cyclic observable and signature matrix change of variables are the same, since 
\begin{subequations}
\begin{align}
    \cycla{i}{\alpha} & = \proja{i}{0}{\alpha} - \proja{i}{1}{\alpha} = 2\proja{i}{0}{\alpha} - 1 = \siga{i}{0}{\alpha},\\
    \siga{i}{0}{\alpha} & = - \siga{i}{1}{\alpha}.
\end{align}
\end{subequations}
In this case we use the simplified notation $\xora{i}{\alpha} = \cycla{i}{\alpha} = \siga{i}{0}{\alpha}$.

\sssec{Strategies as Representations of the Universal Game Algebra}
\label{sssec:Strategies_as_representations}

Recall from \Cref{sec:Math_Background} that a $*$-representation $\rep: \uGA \rightarrow \cB(\cH)$ is an algebraic structure preserving map, so images of generators of $\uGA$
consist of operators on $\cH$ which respect the algebraic structure of $\uGA$, for example
$$
\pi\big((e(\alpha)^i_a)^2\big)=
\pi(e(\alpha)^i_a)^2.
$$

Then we note that any $*$-representation $\rep: \uGA \rightarrow \cB(\cH)$ and state $\state \in \cH$ can be used to define a commuting operator strategy $\strat \in \costrat$, simply by setting 
\begin{align}
    \projr{i}{a}{\alpha} = \rep(\proja{i}{a}{\alpha})
\end{align}
for all $i,a,\alpha$ and taking the Hilbert space $\cH$ to be the target Hilbert space of the $*$-representation~$\rep$. Similarly, any commuting operator strategy defined by a Hilbert space $\cH$, state $\state$, and projectors $\projr{i}{a}{\alpha}$ can be used to define a $*$-representation $\rep: \uGA \rightarrow \cB(\cH)$ by setting 
\begin{align}
    \rep(\proja{i}{a}{\alpha}) = \projr{i}{a}{\alpha}
\end{align}
for all $i,a,\alpha$. For this reason we view commuting operator strategies interchangeably with pairs $(\rep, \state)$ consisting of $*$-representations $\pi:\uGA \rightarrow \cB(\cH)$ and states $\state \in \cH$. 

\sssec{An Algebraic Definition of the Commuting Operator Value of a Game}

For any game $\game$ define the game polynomial $\Phi_\game \in \uGA$ by 
\begin{align}\label{eq:gamepoly}
    \Phi_\game = \sum_{\vec{i} \in [n]^k}  \sum_{\vec{a} \in [m]^k} \mu\left(\vec{i}\right) V\left(\vec{i}, \vec{a}\right) \prod_{\alpha \in [k]} \proja{i(\alpha)}{a(\alpha)}{\alpha} 
\end{align}
Then, recalling the view of strategies as representations introduced in \Cref{sssec:Strategies_as_representations},
\begin{align}\label{eq:covdefn}
    \covalue(\game) = \sup_{\pi, \rho} \tr[\pi(\Phi_\game)\rho]
\end{align}
where the supremum is taken over all $*$-representations $\pi$ of $\uGA$ into bounded operators on a Hilbert space $\cH$, and density operators $\rho \in \cB(\cH)$.

\bs 

A slightly different representation of $\covalue$ which is useful is given in the next lemma.
\begin{lem}
The sup in \Cref{eq:covdefn} is attained for some 
pair representation-density operator $(\hat \pi,\hat \rho)$.
Moreover, one can take 
$\hat \rho$ to be a pure state 
$\hat \rho= \hat \psi^* \hat \psi$,
so
\beq\label{eq:covPsi}
  \covalue(\game) = \max_{\pi,\psi}\psi^*\pi(\Phi_\game) \psi.
\eeq
Here $\psi \in \cH$ is a unit vector, and $\cH$ is the Hilbert space into which $\pi$ represents;
without loss of generality $\cH$ can be taken to be separable.

\end{lem}

\begin{proof}
Firstly, it is well-known that
the extreme points of the convex set of
 density operators
$\rho\in\cB(\cH)$ are exactly rank ones.
This shows one can replace $\rho$ 
by $ \psi^*  \psi$ for unit vectors $\psi$
in \Cref{eq:covdefn}.

Given $\pi,\psi$ as above, the map 
\beq\label{eq:ellGNS}
\ell:\uGA\to\CC, \quad a\mapsto \psi^*\pi(a)\psi
\eeq 
is a positive
linear functional with $\ell(1)=1$. 
Conversely, a normalized positive linear functional $\ell:\uGA\to\CC$ yields through
the Gelfand-Naimark-Segal (GNS) construction a
representation $\pi$ and unit vector $\psi$ such that
\Cref{eq:ellGNS} holds. Thus \Cref{eq:covdefn}
can be rewritten as
\beq\label{eq:covState}
\covalue(\game) = \sup_{\ell} \ell(\Phi_\game).
\eeq
The set of the normalized positive linear functionals on $\uGA$ is weak-$*$ compact
by the Banach-Alaoglu theorem.
Thus by continuity, the supremum in \Cref{eq:covState} is attained.
Hence the GNS construction as explained above yields \Cref{eq:covPsi}.
Observe that since $\uGA$ is countably dimensional, the constructed
Hilbert space $\cH$ will be separable.
\end{proof}

\sssec{Valid and Invalid Response Sets}
\label{sssec:valid_and_invalid_response_sets}
Now we introduce a few more objects which are useful for describing nonlocal games.

First, for any game $\game$, let the  set $\questions \subseteq [n]^k$ contain all question vectors which can be sent to the players with nonzero probability. Formally, 
\begin{align}
    \questions = \{\vec{i} \in [n]^k : \mu(\vec{i}) > 0\}. 
\end{align}
Then, for each $\vec{i} \in \questions$, let $\gr(\vec{i})$ list the response vectors to each question vector which achieve a score of exactly 1, so 
\begin{align}
    \gr(\vec{i})=\{ \vec{a} : \score(\vec{i},\vec{a}) = 1\}.
\end{align} 
We say the set $\gr(\vec{i})$ contains all the valid or ``winning'' responses to the question vector $\vec{i}$. 
We also define the set of invalid responses or ``losing'' responses $\br(\vec{i})$ to be the complement of the set $\gr(\vec{i})$.

In this paper we assume that all games have a scoring function $\score$ whose image is either $0$ or $1$ and a distribution $\mu$ which is uniform over a set of allowed questions. For these games, specifying a question set $\questions$ and sets $\gr(\vec{i})$ (or $\br(\vec{i})$) for all $\vec{i} \in \questions$ completely specifies the game, since we can write 
\begin{align}
    \Phi_\game = \frac{1}{\abs{\questions}} \sum_{\vec{i} \in \questions} \sum_{\vec{a} \in \gr(\vec{i})} \prod_{\alpha \in [k]} \proja{i(\alpha)}{a(\alpha)}{\alpha}.
\end{align}

\ssec{Determining Perfect Commuting Operator Strategies for Games}

\label{ssec:Characterizing_perfect_games}

A commuting operator strategy $\strat = (\rep,\state)$ is called a perfect strategy for a game $\game$ if $\score(\game,\strat) = 1$. In this section, we develop some terminology for describing perfect commuting operator strategies for nonlocal games.

We first introduce a general definition which will be used frequently in later sections. 

\begin{defn} \label{defn:dset}
A game $\game$ is said to be \df{\dby} a set of elements $\dset \subseteq \uGA$ (or, equivalently, $\dset$ is said to be a determining set of $\game$) if it is true that any strategy $\strat = (\rep,\state)$ satisfies $\score(\game,\strat) = 1$ iff 
$\rep(\delt)\state = 0$ for all $\delt \in \dset$. 
\end{defn}

\sssec{Determining Sets of a Game}
\label{sssec:valid_invaid_response_det_sets}

We next show that any game $\game$ has two natural \detsets, based on the valid and invalid response sets introduced in \Cref{sssec:valid_and_invalid_response_sets}. We first define these sets, then show they are both determining sets for the game~$\game$. 

\begin{defn} \label{defn:valid_invalid_poly}
For any game $\game$ with question set $\questions$ and valid responses $\gr(\vec{i})$, we introduce a companion set
of valid elements of $\uGA$,
\begin{align}
    \validpoly := \left\{ \sum_{\vec{a} \in \gr(\vec{i})} \prod _ \alpha \proja{i(\alpha)}{a(\alpha)}{\alpha}  -1 \right\}_{\vec{i} \in \questions},
\end{align}
Similarly, define the invalid elements $\invalidpoly$ by 
\begin{align}
    \invalidpoly := \left\{ \prod _ \alpha \proja{i(\alpha)}{a(\alpha)}{\alpha} \right\}_{(\vec{i},\vec{a}) \in \questions \times \br(\vec{i})}.
\end{align}
\end{defn}

\begin{thm} \label{thm:perfect_games_direc_zero}
Let $\game$ be a game and $\validpoly$, $\invalidpoly$ be  as in \Cref{defn:valid_invalid_poly}. Then $\game$ is \dby both $\validpoly$ and $\invalidpoly$.
\end{thm}

\def\psd{positive semidefinite}

\begin{proof}
By definition, a strategy $(\rep,\state)$ for a nonlocal game is perfect iff 
\begin{align}
    \state^* \rep(\Phi_G)\state &= 1,
\end{align}
i.e.,
\begin{align}  
 \state^* \frac{1}{\abs{\questions}} \sum_{\vec{i} \in \questions}  \sum_{\vec{a} \in \gr(\vec{i})} \prod_{\alpha \in [k]} \rep\left(\proja{i(\alpha)}{a(\alpha)}{\alpha}\right) \state &= 1.
\end{align}
This equation has the form: the average of a function equals its maximum, so each term equals the maximum. We now  exploit this:
 for all $\vec{i} \in \questions$ we have
\begin{align}\label{eq:someGPineq}
    \rep\left(\sum_{\vec{a} \in \gr(\vec{i})} \prod_{\alpha \in [k]} \proja{i(\alpha)}{a(\alpha)}{\alpha}\right)  \leq \rep\left(\sum_{\vec{a} \in [m]} \prod_{\alpha \in [k]} \proja{i(\alpha)}{a(\alpha)}{\alpha}\right) = I
\end{align} 
hence 
\begin{align}
    \state^* \rep\left(\sum_{\vec{a} \in \gr(\vec{i})} \prod_{\alpha \in [k]} \proja{i(\alpha)}{a(\alpha)}{\alpha}\right) \state \leq \state^* \state = 1
\end{align} 
and a game is perfect iff we have for all $\vec{i} \in \questions$:
\begin{align}\label{eq:trace_proj_direc_zero0}
    \state^* \rep\left(\sum_{\vec{a} \in \gr(\vec{i})} \prod_{\alpha \in [k]} \proja{i(\alpha)}{a(\alpha)}{\alpha} \right) \state &= 1,
\end{align}
equivalently,
\begin{align}\state^* \rep \left(\left(\sum_{\vec{a} \in \gr(\vec{i})} \prod_{\alpha \in [k]} \proja{i(\alpha)}{a(\alpha)}{\alpha} \right) - 1 \right) \state &= 0. \label{eq:trace_proj_direc_zero}
\end{align}
 Again using \Cref{eq:someGPineq},
we see 
\begin{align}
    \rep\left(\sum_{\vec{a} \in \gr(\vec{i})} \prod_{\alpha \in [k]} \proja{i(\alpha)}{a(\alpha)}{\alpha} \right) - I
\end{align}
is negative semidefinite, hence \Cref{eq:trace_proj_direc_zero} implies 
\begin{align}
    \rep\left(\left(\sum_{\vec{a} \in \gr(\vec{i})} \prod_{\alpha \in [k]} \proja{i(\alpha)}{a(\alpha)}{\alpha} \right) - I \right) \state = 0 
\end{align}
which finishes the first part of the proof.

To convert this condition from terms of $\gr$ to terms of $\br$ fix $\vec{i} \in \questions$ and use

\begin{align}
    \sum_{\vec{a} \in \gr(\vec{i}) \cup \br(\vec{i})   } \state^* \rep \left(\prod_{\alpha \in [k]} \proja{i(\alpha)}{a(\alpha)}{\alpha} \right) \state = 
      \state^* \state = 1 . 
\end{align}
In words we are summing over all responses valid or invalid to question $\vec{i}$.
Subtract \Cref{eq:trace_proj_direc_zero0} from this to get
$\sum_{\vec{a} \in  \br(\vec{i})   } =0 $.
This is a sum of nonnegative terms, so each is 0:
\begin{align}
    \state^* \rep \left(\prod_{\alpha \in [k]} \proja{i(\alpha)}{a(\alpha)}{\alpha}\right) \state 
    = 0  \ \ \  \forall \; \vec{a} \in \br(\vec{i}) 
    \label{eq:perfectGameInvalidResponse2}
    \end{align}
Since operators 
    $\rep\left(\proja{i}{a}{\alpha}\right)$ are \psd \ 
    we get 
\[
    \rep \left(\prod_{\alpha \in [k]} \proja{i(\alpha)}{a(\alpha)}{\alpha}\right) \state = 0 \ \ \ \forall \; \vec{a} \in \br(\vec{i}) ,
\]
    thus finishing the proof.
\end{proof}

\sssec{Torically Determined Games}
\label{sssec:tgames}

While any game $\game$ is determined by the sets $\validpoly$ and $\invalidpoly$, there are, in general, many other sets of elements in $\uGA$ that determine a game $\game$. An important class of games in this paper are \tgames, which we define as games \dby a particularly nice set of binomial elements.

\begin{defn}
A game $\game$ is called a \df{\tgame} if there exists a group $G$ with $\uGA \cong \mathbb{C}[G]$ and $\game$ is determined by a set of elements 
\begin{align}
    \dset = \{\beta_i g_i - 1\}
\end{align}
with each $\beta_i \in \mathbb{C}$ and $g_i \in G$. In this case we say $\game$ is \df{\tdby} the set $\dset$ and call the elements $\beta_i g_i$ clauses of $\dset$. 
\end{defn}

Traditionally, the term toric ideal refers to ideals generated by binomials,  these being  the 
difference of two monomials. 
The next lemma shows that our use of the term toric is consistent with this.
\begin{lem}
A game $\game$ which is determined by a set of elements of the form 
\begin{align}
    \dset' = \{\beta_i g_i - \beta_i' g_i'\}_i
\end{align}
with all $\beta_i, \beta_i' \in \mathbb{C}$, and $g_i, g_i' \in G$ for some group $G$ with $\uGA \cong \mathbb{C}[G]$ is also {torically determined}.
\end{lem}

\begin{proof}
For all $(\rep,\state)$ we have 
\begin{align}
    \rep(\beta_i g_i - \beta_i' g_i') \state = 0 \qquad  
    \Leftrightarrow  \qquad 
    (\; \rep(\beta_i (\beta_i')^{-1} g_i (g_i')^{-1} - 1) \state &= 0 
\end{align}
so $\game$ is also (torically) determined by the set 
\[
    \dset = \{\beta_i (\beta_i')^{-1} g_i (g_i')^{-1} - 1\}_i.
\qedhere\]
\end{proof}

To provide some further intuition about the definition of \tgames we introduce the concept of a \relaxgp. Given a game $\game$ with game polynomial $\Phi_\game$, we define a \df{\relaxgp} for the game to be any element $\rgp_\game\in\uGA$ with the property that for any representation $\rep$ and state $\state$ such that
\begin{align}
    \rep(\gp_\game) \state = \state
\end{align}
we have
\begin{subequations}
\begin{align}
    \rep(\rgp_\game) \state &= \state, \\
    \abs{\phi^*\rep(\rgp_\game) \phi } &\leq 1 \qquad \text{for all }\ 
    \| \phi  \| \leq 1 .
\end{align}
\end{subequations}

Informally, scoring a strategy using a  \relaxgp produces the correct score if a strategy is perfect for the associated game, but may give the wrong score otherwise. In particular, a \relaxgp need not be self-adjoint, and the ``score'' coming from the \relaxgp may not even be real when a strategy is not perfect. 

Then we note that a game $\game$ which is \tdby a set of elements $\dset$ with clauses $\clauseset$ has a relaxed game polynomial of the form 
\begin{align}
    \rgp_\game = \frac{1}{\abs{\clauseset}} \sum_{\clause \in \clauseset} \clause.
\end{align}
Similarly, any game $\game$ with a relaxed game polynomial 
\begin{align}
    \rgp_\game = \frac{1}{\abs{\clauseset}} \sum_{\clause \in \clauseset} \clause
\end{align}
where each $\clause \in \clauseset$ is of the form $\beta_i g_i$ with $\beta_i \in \mathbb{C}$, $\abs{\beta_i} = 1$, and $g_i \in G$ for some group $G$ with $\uGA \cong \mathbb{C}[G]$ is \tdby the set of elements $\dset = \clauseset - 1$. 
Since $\rgp_\game$
is an average,
the proof of both these statements follows from an argument very similar to the proof of \Cref{thm:perfect_games_direc_zero}.

\ssec{XOR and Mod \texorpdfstring{$r$}{r} Games}

Now we practice using some of the machinery introduced in the previous section to describe XOR games, and a natural generalization which we call Mod $r$ games. In this section we often describe games in a variety of ways (for example using game polynomials, valid or invalid response sets, and \relaxgps) without proving that these definitions are equivalent. In all cases the proof of equivalence amounts to a routine calculation using the definitions given in \Cref{ssec:Algebraic_Picture}.

\label{ssec:Games_Examples}
\sssec{XOR Games}
\label{ssec:Group_Algebra_Transl_for_XOR_Games}

\noindent 

XOR games are games with $m = 2$ responses which we interpret as a $0$ or a $1$.
The  valid responses $\gr(\vec{i})$ 
to each question vector $\vec{i}$ are all responses which sum to either $0$ or sum to $1$ mod $2$. 
We can think of an XOR game with $T$ possible questions, labeled $\vec{i}_t = (i_t(1), i_t(2), ... , i_t(k))$ for $t \in [T]$ as testing the satisfiability of a system of $T$ equations, where each equation takes the form
\begin{align}
    \sum_{\alpha \in [k]} y^{(\alpha)}_{i_t(\alpha)} = s_t \pmod{2},
\end{align}
the $y^{(\alpha)}_{i_t(\alpha)}$ are free variables taking values in $\{0,1\}$, and each $s_t \in \{0,1\}$ specifies the winning parity associated with each question vector $\vec{i}_t$.

The game polynomial of an XOR game takes the form 
\begin{align}
    \Phi_\game = \frac{1}{2} + \frac{1}{2T} \sum_{t=1}^T  (-1)^{s_t} \prod_{\alpha \in [k]} \xora{i_{t}(\alpha)}{\alpha}
\end{align}
with $T > 0$ some integer, the vector $\vec{i}_t \in [n]^k$ and the integer $s_t \in \{0,1\}$ are arbitrary and the notation $i_{t}(\alpha)$ refers to the $\alpha$-th entry of the vector $\vec{i}_t$. 
In anticipation of a connection to \tgames, we refer to each monomial
\begin{align}
(-1)^{s_t} \prod_{\alpha \in [k]} \xora{i_{t}(\alpha)}{\alpha}
\end{align} as a clause, so the game polynomial above corresponds to a $T$-clause XOR game.  

When working with XOR games it is convenient to remove the constant factor in the game polynomial and rescale the remaining terms, producing a \relaxgp 
\begin{align}
    \rgp_\game = \frac{1}{T} \sum_{t=1}^T  (-1)^{s_t} \prod_{\alpha \in [k]} \xora{i_{t}(\alpha)}{\alpha}.
\end{align}
This \relaxgp computes a quantity known as the bias (meaning $\tr[\rep(\rgp_\game)\rho]$ gives the bias achieved by the strategy defined by $\rep, \rho$). For this reason we also call the \relaxgp for XOR games defined in this way the \emph{bias polynomial} of an XOR game. Note that an XOR game can be completely specified by describing its bias polynomial. We also note that the bias polynomial formulation of XOR games makes it immediately clear that XOR games are \tgames, by the discussion in \Cref{sssec:tgames} and the observation that elements $\xora{i_{t}(\alpha)}{\alpha}$ are also cyclic unitary generators as defined in \Cref{ssec:Algebraic_Picture}.

\sssec{Mod \texorpdfstring{$r$}{r} Games} 

\label{sssec:Mod $r$ Games}
Mod $r$ games are a natural generalization of XOR games where the players have $m$ responses and the valid responses to any question vector $\vec{j}_t$ are responses which satisfy some linear equation mod $r$. Formally, we specify a $T$ question Mod $r$ game by a system of equations
\begin{align}
    \left\{\sum_{\alpha \in [k]} d_t(\alpha) y^{(\alpha)}_{j_t(\alpha)} = s_t \pmod{r}\right\}_{t \in [T]}.
\end{align}
The $d_t(\alpha)$ and $s_t$ are integers in $[r]$, while the 
$y^{(\alpha)}_{j_t(\alpha)}$ are free variables. Players' responses $a(1), a(2), \ldots , a(k)$ to question vector $\vec{j}_t$ are winning if setting $y^{(\alpha)}_{j_t(\alpha)} = a(\alpha)$ satisfies the $t$-th equation in the system of equations. We show in \Cref{thm:Mod(r)_games} that Mod $r$ games are \tdby a set of elements of the form 
\begin{align}
    \dset = \left\{(\exp(- 2\pi i /r))^{s_t} \prod_{\alpha \in [k]} \left(\cycla{j_{t}(\alpha)}{\alpha}(\alpha)\right)^{d_t(\alpha)} - 1\right\}_{t \in [T]}
\end{align}
and hence admit a \relaxgp of the form 
\begin{align}
    \rgp_\game = \frac{1}{T} \sum_{t=1}^T  (\exp(- 2\pi i /r))^{s_t} \prod_{\alpha \in [k]} \left(\cycla{j_{t}(\alpha)}{\alpha}(\alpha)\right)^{d_t(\alpha)},
\end{align}
with the $\cycla{j_{t}(\alpha)}{\alpha}$ cyclic unitaries of order $m$, as described in \Cref{sssec:uGA}. We also note that a very similar result (\Cref{thm:LSGames_Toric}) is proven for linear systems in \cref{sec:LinearSysGames}.

\begin{thm}\label{thm:Mod(r)_games} A Mod $r$ game specified by a system of equations 
\begin{align}
    \left\{\sum_{\alpha \in [k]} d_t(\alpha) y^{(\alpha)}_{j_t(\alpha)} = s_t \pmod{r}\right\}_{t \in [T]}
\end{align}
is torically determined by a set of elements of the form 
\begin{align}
    \dset = \left\{(\exp(-2\pi i /r))^{s_t} \prod_{\alpha \in [k]} \left(\cycla{j_{t}(\alpha)}{\alpha}(\alpha)\right)^{d_t(\alpha)} - 1\right\}_{t \in [T]}
\end{align}
\end{thm}

\begin{proof}
First we write cyclic unitary generators in terms of projectors $\proja{i}{a}{\alpha}$ and expand out the resulting product to note 
\begin{align}
   \prod_{\alpha \in [k]} \left(\cycla{j_{t}(\alpha)}{\alpha}(\alpha)\right)^{d_t(\alpha)} = \sum_{\vec{a} \in [r]^k} \exp( \frac{2 \pi i}{r} \sum_{\alpha \in [k]} d_t(\alpha) a(\alpha) ) \prod_{\alpha} \proja{j_t(\alpha)}{a(\alpha)}{\alpha}. \label{eq:expanding_mod(r)_cylca_product}
\end{align}
Then we define 
\begin{align}
    A(t) = \left\{ \vec{a} : \sum_{\alpha \in [k]} a(\alpha)d_t(\alpha) = s_t \right\}
\end{align}
to be the collection of response vectors which win on
the question corresponding to clause $t \in [T]$ of the mod $r$ game. Then, \Cref{eq:expanding_mod(r)_cylca_product} gives that, for any commuting operator strategy $(\rep,\state)$ the condition 
\begin{align}
    \rep\left(\prod_{\alpha \in [k]} \left(\cycla{j_{t}(\alpha)}{\alpha}(\alpha)\right)^{d_t(\alpha)}\right)\state = \state
\end{align}
is equivalent to the condition 
\begin{align}
    \rep\left(\sum_{\vec{a} \in  A(t)} \prod_{\alpha} \proja{j_t(\alpha)}{a(\alpha)}{\alpha} \right) \state = \exp(\frac{2 \pi i s_t}{r})\state
\end{align}
and thus the condition 
\begin{align}
    \rep\left(\exp(-2\pi i /r))^{s_t} \prod_{\alpha \in [k]} \left(\cycla{j_{t}(\alpha)}{\alpha}(\alpha)\right)^{d_t(\alpha)}\right) \state = \state
\end{align}
ensures the players' response is always winning for question $t \in [T]$. The result follows. 
\end{proof}
\sssec{Example}
\label{sssec:Perfect_criterion_example_GHZ}

To provide a concrete example of the various ways of characterizing perfect games we discuss the GHZ game. This is a 3 player XOR game with bias polynomial
\begin{align}
    \Phi_{GHZ} = \frac{1}{4}\left(\xora{0}{1}\xora{0}{2}\xora{0}{3} - \xora{1}{1}\xora{1}{2}\xora{0}{3} - \xora{1}{1}\xora{0}{2}\xora{1}{3} - \xora{0}{1}\xora{1}{2}\xora{1}{3}\right)
\end{align}
Then, equivalently, the GHZ game is determined by the set of elements 
\begin{align}
    \{\xora{0}{1}\xora{0}{2}\xora{0}{3} - 1, -\xora{1}{1}\xora{1}{2}\xora{0}{3} - 1, -\xora{1}{1}\xora{0}{2}\xora{1}{3} - 1, -\xora{0}{1}\xora{1}{2}\xora{1}{3} - 1 \}
\end{align}and has a perfect commuting operator strategy iff there exists a Hilbert space $\cH$, state $\state \in \cH$, and a $*$-representation $\rep: \uGA \rightarrow \cB(\cH)$ satisfying 
\begin{subequations}
\begin{align}
    \rep\left(\xora{0}{1}\xora{0}{2}\xora{0}{3}\right)\state &=  \state, 
    \\
    \rep\left(\xora{0}{1}\xora{1}{2}\xora{1}{3}\right)\state = \rep\left(\xora{1}{1}\xora{0}{2}\xora{1}{3}\right)\state = \rep\left(\xora{1}{1}\xora{1}{2}\xora{0}{3}\right)\state &= -\state.
\end{align}
\end{subequations}
Such a representation can be found by taking $\cH$ to be an eight dimensional Hilbert space, $\psi$ a 3 qubit GHZ state, and letting $\pi$ map elements $\xora{i}{\alpha}$ to the standard measurement operators for the GHZ game or, explicitly
\begin{subequations}
\begin{align}
    \pi\left(\xora{0}{1}\right) &= \sigma_X \otimes I \otimes I &
    \pi\left(\xora{1}{1}\right) &= \sigma_Y \otimes I \otimes I \\
    \pi\left(\xora{0}{2}\right) &=  I \otimes \sigma_X \otimes I &
    \pi\left(\xora{1}{2}\right) &= I \otimes \sigma_Y \otimes I \\
    \pi\left(\xora{0}{3}\right) &= I \otimes I \otimes \sigma_X  &
    \pi\left(\xora{1}{3}\right) &= I \otimes I \otimes \sigma_Y  
\end{align}
\end{subequations}
where $\sigma_X$ and $\sigma_Y$ are the Pauli $X$ and $Y$ operators, respectively, and $I$ denotes a dimension 2 identity matrix. 

While it is certainly easiest to describe the GHZ game using the game polynomial formulation, we can also describe perfect strategies for the game using the language of valid and invalid response sets. Using this language the question set of the  GHZ game is given by $\questions_{GHZ} = \{(0,0,0),(0,1,1),(1,0,1),(1,1,0)\}$ with valid response sets 
\begin{subequations}
\begin{align}
    \gr_{GHZ}(0,0,0) &= \{\textit{EVEN}\}, 
    \\
    \gr_{GHZ}(0,1,1) = \gr_{GHZ}(1,0,1) = \gr_{GHZ}(1,1,0) &=  \{\textit{ODD}\} ,
\end{align}
\end{subequations}
where $\{\textit{EVEN}\}$ and $\{\textit{ODD}\}$ denote the set of all 3 bit response strings containing an even and odd number of ones, respectively.
The invalid response sets are described similarly, but with $\{\textit{EVEN}\}$ and $\{\textit{ODD}\}$ swapped.

We can obtain projectors corresponding to a perfect strategy for the GHZ game using the same representation as above, with
\begin{subequations}
\begin{align}
    \projr{0}{0}{1} &= \pi\left(\proja{0}{0}{1}\right) = \pi\left(\frac{1}{2}\left(1 + \xora{0}{1} \right)\right), \\ 
    \projr{0}{1}{1} &= \pi\left(\proja{0}{1}{1}\right) = \pi\left(\frac{1}{2}\left(1 - \xora{0}{1} \right)\right),
\end{align}
\end{subequations}
and so on. We leave it as an exercise for the reader to check that the projectors defined in this way satisfy the valid and invalid response perfect game characterization conditions laid out in \Cref{thm:perfect_games_direc_zero}.

\sec{\NullSSs for Perfect Nonlocal Games}
\label{sec:PerfectGamesNullSS}

\NullSSs are
algebraic certificates for
polynomial equations to be solvable
or, even stronger, for one set of polynomial equations to have solutions contained in 
the set of solutions to another.
Solutions to a polynomial equation such as $X^2 =I$ are often called zeroes of the polynomial 
$p(x):= x^2 -1$.

   \ssec{Three types of zeroes}
   
   \def\zd{{\mathcal Z}_{\textrm{dir}} }
   \def\zh{{\mathcal Z}_{\textrm{hard}} }
   
   For nc polynomials there are three natural types of zero: hard zeroes, directional zeroes and determinantal zeros. 
   These have been studied in 
   the mathematics community for several decades.

   Let us illustrate 
   directional zeroes since that is what we work with most in this paper. Given $f \in \fa$ 
   define the \df{directional zero set}
    \begin{align}
    \zd (f):= &\{(X, \psi)\mid X_j \in \cB(\cH), \ 0\neq \psi \in  \cH  , 
    \ f(X)\psi=0    \text{ over all }  \cH \}. 
    \end{align}
    In the language of representations we have
     \begin{equation}
    \zd (f):=
    \{ (\pi, \psi) \mid
    \pi(f)\psi=0, \ \pi:\fa\to\cB(\cH) \text{ representation}, \ 
     0\neq\psi \in \cH \}.
    \end{equation}
    Quantum strategies need refinements of this set up; for one,
    we have not captured
    crucial algebraic relationships such as
    our whole world lives inside an algebra.
    Thus for an arbitrary algebra $\cA$ and $f\in\cA$ we define
    \beq 
    \zd^{\cA} (f)=
    \{(\pi, \psi)
\mid
     \pi(f)\psi=0, \ \pi:\cA\to\cB(\cH) \text{ representation}, \
    0\neq \psi \in \cH \}.
    \eeq 
        Essential to strategies is the intersection of this with $*$-representations in the context of $*$-algebras $\cA$:
    \begin{align}
    \zd^{\mathrm{re}, \cA} (f) =   \zd^{ \cA} (f) \cap 
    \{(\pi,\psi)\mid \pi \text{ $*$-representation}\}.
    \end{align}
   
    So far we have motivated and given the
    definition of the set
    of directional zeroes.
    With motivation in the same spirit, we define hard and determinantal real zeros by
     \[
     \begin{split}
    \cZ_{\textrm{hard}}^{ \cA} (f)&=
   \{\pi
\mid
     \pi(f)=0, \ \pi:\cA\to\cB(\cH) \text{ representation} \} \\
    \cZ_{\det}^{\cA} (f) & =
    \{\pi \mid \det[\pi(f)]=0, \
    \pi:\cA\to\cB(\cH) \text{ representation with } \dim\cH<\infty \}.
    \end{split}
    \]
    Directional and hard zeroes (cf.~\Cref{sec:Synchronous}) are central to quantum games.
    The definition of determinantal zeros is included  here for perspective only. 
    
    Of course 
    one also has the $*$-representation version
    of these:
    $$
     \cZ_{\textrm{hard}}^{\textrm{re}, \cA} \qquad \text{and} \qquad \cZ_{\det}^{\textrm{re}, \cA}
     $$
    Also for a subset $F \subseteq \fa$  we define 
    $$\cZ(F)=  \ \bigcap_{f \in F} \cZ(f) .$$

\sssec{Example}

In \Cref{ssec:Characterizing_perfect_games} we described perfect commuting operator strategies for nonlocal games as $*$-representations of an algebra $\uGA$ which satisfied certain desiderata 
which we now re-describe
using the language of nc polynomial zeros introduced in the section above. For concreteness, we consider the GHZ game introduced in \Cref{sssec:Perfect_criterion_example_GHZ}. 

As discussed in \Cref{sssec:Perfect_criterion_example_GHZ}, the GHZ game is determined by a set of polynomials
\begin{align}
    \dset = \{\xora{0}{1}\xora{0}{2}\xora{0}{3} - 1,\  \xora{0}{1}\xora{1}{2}\xora{1}{3} + 1 ,\  \xora{1}{1}\xora{0}{2}\xora{1}{3} +1,\  \xora{1}{1}\xora{1}{2}\xora{0}{3} + 1 \}.
\end{align}
Equivalently, in the language of directional zeros we see that the GHZ game has a perfect commuting operator strategy iff  the directional zero set $\zd^{\mathrm{re}, \uGA} (\dset)$ is nonempty. Furthermore, entries in this set correspond to representations and states $(\rep,\state)$ defining perfect commuting operator strategies. 

For an example  using hard zeros see  \Cref{sec:GBExamples1}.

 \ssec{A general noncommutative \NullSS}
 \label{ssec:QuantumGamesNullSS}

     \def\cH{\mathcal H}
     \def\hL{{\widehat L}}

Let $\cA$ be a $*$-algebra.
We equip $\cA$ with a certain topology called the finest locally convex topology\footnote{This is  
the finest vector space topology whose every neighborhood of zero contains a convex balanced absorbing set. Equivalently, it is the coarsest
topology for which every seminorm on $\cA$ is continuous. In this case,
every linear functional $f$ on $\cA$ is continuous since $|f |$ is a seminorm.
Hence
for a convex subset $C\subseteq\cA$ we have
\[
\clos(C)=\{
    c\in\cA \mid \forall \text{ linear } \ell:\cA\to\CC \text{ with } \ell(C)\subseteq\R_{\geq0} \text{ we have } \ell(c)\geq0
    \}
\]
by the Hahn-Banach separation theorem.}.
For our purposes one does not need to worry about it since we will soon move to a much less general ``weak'' \NullSS. For a subset
$C \subseteq \cA$ let $\clos(C)$ denote the 
closure of $C$ in the finest locally convex 
topology. We use $\SOS_C$ to denote
all sums of $u^* u$ with $u \in C$.

     Results of Section 5 of
     \cite{CHMN13}, see also
     \cite{CHKMN14},
     adapted to our use case are
          as follows. 
     
     \begin{thm} 
     \label{thm:rnssCHMN}
     Suppose that 
     $\cLI $ is the left ideal of $\cA$ generated by $F\subseteq\cA$.
Then for $a\in\cA$ the following are equivalent:
\ben[\rm(i)]
\item\label{eq:repCond}  $\pi(a)\psi=0$
     for every  $*$-representation 
      $\pi : \cA \rightarrow \cB(\cH)$ for some (possibly infinite dimensional) Hilbert space $\cH$ and vector $\psi\in\cH$ 
      such that $\pi(\pol)\psi=0$
     for all $\pol \in F$; 
     \item[\rm(i)'] $\zd^{\mathrm{re}, \cA}(F)\subseteq \zd^{\mathrm{re}, \cA}(a)$;
     \item 
     \label{eq:genNullss}
      $ 
     -a^* a \in \clos[\SOS_\cA - \SOS_{\cLI}];
      $ 
           \item 
     $
     -a^* a \in \clos(\SOS_\cA - \cone(S) ),
     $
     where $S= \   
     \{f^*f\mid f\in F \}$;
     
     \item 
      $
     -a^* a \in \clos[
     \SOS_\cA + \cLI + \cLI^*].
     $     
\een 

     When 
     $\cA=\fa$ is the free algebra and 
      $\cLI$ is finitely generated,  $*$-representations $\pi$ into finite dimensional Hilbert spaces
       suffice in \Cref{eq:repCond}.
     \end{thm}
     
     \begin{cor}
     \label{cor:2sidedI}
     Suppose $\ide\subseteq\cA$ is a $*$-closed two sided ideal.
     Then the following are equivalent for $a\in\cA$:
     \ben[\rm(i)]
     \item
     $\pi(a)=0$
     for every $*$-representation 
     $\pi$ such that $\pi(f)=0$
     for all $f \in \ide$;
     \item[\rm(i)']
    $\zh^{\mathrm{re}, \cA}(\ide)\subseteq\zh^{\mathrm{re}, \cA}(a)$;
\item\label{it:aaa}
     $
     -a^* a \in \clos[ \SOS_\cA + \ide].$
     \een
 \end{cor}
 
 \begin{proof}    
 Items (i) and (i)' are equivalent by the definition on $\zh^{\mathrm{re}, \cA}$.

In the context of  \Cref{thm:rnssCHMN},
$\cLI:=\ide=\cLI^*$ is a left (and right) ideal,
so the algebraic certificate in \Cref{it:aaa} is
\beq 
     -a^* a \in \clos [ \SOS_\cA + \ide ]
     = \clos [ \SOS_\cA + \cLI +\cLI^*].
\eeq 
Thus if \Cref{it:aaa} fails to hold, then
by \Cref{thm:rnssCHMN} there is a 
$*$-representation 
$\pi:\cA\to\cB(\cH)$ and  $0\neq \psi\in\cH$ with
with $\pi(\ide)\psi=\{0\}$,
and $\pi(a)\psi\neq0$.

Now consider the Hilbert space $\check\cH:=[\pi(\cA)\psi]\subseteq\cH$.
By definition, $\pi(\cA)\check\cH\subseteq\check\cH$, so $\pi$ induces
a $*$-representation $\check\pi:\cA\to\cB(\check\cH)$.
By construction, $\check\pi(\ide)=\{0\}$ (this uses that $\ide$ is a right ideal, too), but
$\check\pi(a)\psi\neq0$, so $\check\pi(a)\neq0$.
\end{proof}

\begin{proof}[Proof of \Cref{thm:rnssCHMN}]
Here we give a proof which ties the parts of the theorem to corresponding theorems in \cite{CHMN13}. This is unintuitive so in 
 \Cref{sec:pfideaNulSS}
we sketch the idea of the proof.

\Cref{eq:repCond} and Item (i)' are equivalent by the definition on $\zd^{\mathrm{re}, \cA}$.
The equivalence between \Cref{eq:repCond} and \Cref{eq:genNullss}
is  \cite[Theorem 5.1]{CHMN13}. 
The equivalence \ref{eq:repCond} $\Leftrightarrow$ (iv) is
 \cite[Corollary 5.3]{CHMN13}. 
Finally, (iii) $\Leftrightarrow$ \ref{eq:genNullss} now follows from
\cite[Proposition 5.2]{CHMN13}.

The finite dimensional assertion is proved constructively and this is the major part of \cite{CHMN13}; see \cite[Proposition 6.8 and Theorem 2.1]{CHMN13}.
\end{proof}

\begin{rmk}
In general closures in \Cref{thm:rnssCHMN} and \Cref{cor:2sidedI}
are needed as can be shown by employing the Weyl algebra.
\end{rmk}

\sssec{Intuition behind the proof of \texorpdfstring{\Cref{thm:rnssCHMN}}{Theorem 2.2}}
\label{sec:pfideaNulSS}
     
     Here is a special case of
     \Cref{thm:rnssCHMN},
     included here since a sketch of its proof  supplies the readers intuition.
     Also this lesser level of generality is all that is needed here for quantum games, 
     so this theorem is what gets referenced later.

     \newcommand{\langl}{\begin{picture}(4.5,7)
\put(1.1,2.5){\rotatebox{60}{\line(1,0){5.5}}}
\put(1.1,2.5){\rotatebox{300}{\line(1,0){5.5}}}
\end{picture}}
\newcommand{\rangl}{\begin{picture}(4.5,7)
\put(.9,2.5){\rotatebox{120}{\line(1,0){5.5}}}
\put(.9,2.5){\rotatebox{240}{\line(1,0){5.5}}}
\end{picture}}

     \begin{thm}
     \label{thm:AB}
     Suppose 
     \ben [\rm(1)]
     \item
     $\cA$ is a $*$-algebra, where $\SOS_{\cA}$ is \df{archimedean} in the sense that for every $a\in\cA$ there is $\eta\in\mathbb N$ with $\eta-a^*a\in\SOS_{\cA}$;\footnote{This notion of archimedean should not be confused with the algebra $\cA$ being \df{archimedean closed}, meaning that for any $a \in \cA$ with $a + \epsilon \in \SOS_{\cA}$ for each $\epsilon \in \mathbb{R}_{>0}$ we also have $a \in \SOS_{\cA}$.}
     \item $\cLI\subseteq\cA$ is a left ideal.
     \een 
      Then 
the following are equivalent:
\ben[\rm(i)]
\item
      there exist a 
      $*$-representation $\pi:\cA\to\cB(\cH)$
      and  $0\neq\psi \in \cH$ satisfying
      \beq 
      \label{eq:fX0}
      \pi(f) \psi = 0  
      \eeq 
      for all $f \in \cLI$;
\item[\rm(i)']
$\zd^{\mathrm{re}, \cA}(\cLI)\neq\varnothing$;
      \item $-1 \not \in \SOS_{\cA} + \cLI + \cLI^*.$
      \een
     \end{thm}
     
     \begin{proof} 
      As before, Items (i) and (i)' are equivalent by the definition on $\zd^{\mathrm{re}, \cA}$. We thus establish (i) $\Leftrightarrow$ (ii).

     Easy side:  suppose  $-1  \in \SOS_{\cA} + \cLI + \cLI^*$.
     If $\pi,\psi$ as in  \Cref{eq:fX0} exist, then 
     $- \psi^*  \psi \geq 
      \psi^* \SOS_\cA \psi \geq 0$; contradiction.
     
     Harder side: suppose  $-1  \not\in \SOS_{\cA} + \cLI + \cLI^*$.
By the Hahn-Banach theorem (version due to
     Eidelheit-Kakutani \cite[Theorem III.1.7]{barvinok}) 
     there is a
     linear functional $L:\cA\to\C$ satisfying
      \beq\label{eq:2GNS}  L(1)= 1, \qquad L(\SOS_{\cA} + \cLI + \cLI^*) \subseteq \R_{\geq0}.
      \eeq 
{We remark that the strict separation is automatic since by archimedeanity
$L(1)\neq0$.}

     Since $\cLI$ is a subspace, the second property of \Cref{eq:2GNS} implies $L(\cLI)=\{0\}$. Likewise,  
     $L(f)\geq 0$ for any $f\in\SOS_{\cA}$.
We remark that $L(f)^* = L(f^*)$ for all $f\in\cA$. Indeed,
     because $\cA$ is archimedean, every self adjoint
     $g=g^*\in\cA$ is bounded above: there is $\eta\in\mathbb N$ with
     $\eta-g\in\SOS_{\cA}$. Thus $L(g)\in\R$. Then we write $h\in \cA$
     as a linear combination of self adjoints,
     \[
h=\frac{h+h^*}2+i\frac{h-h^*}{2i},
     \]
     to get $L(h^*)=L(h)^*$.
     
     Now perform the GNS construction.
   Define the bilinear
     form  
     \beq\label{eq:preDP} \langl a \mid b\rangl:= L(b^*a)\eeq on
     $\cA$.
     Set
     $N:=\{ a\in\cA \mid  L(a^*a) \}= 0$.
By the Cauchy–Schwarz inequality for semi-scalar products,
     $$0 \leq L(a^*r^* ra)^2 \leq
     L(a^*a)
     L(a^*  r^*r  r^*r a) =0,
     $$
     so $N$ is a left ideal.
     Since $1\not\in N,$
      $N \not = \cA$.
Form the quotient space $\check\cH=\mathcal A/N$. Then \Cref{eq:preDP}
induces a scalar product on $\check\cH$. We complete it to the Hilbert space $\cH$. 
Let 
    $ \phi: \cA \to \cH$,
     $$\phi(a):= a + N$$ be the quotient map and let $ \psi:= \phi(1).$

     Define a $*$-representation
     $\pi$ of $\cA$ on $\cH$ by
     $$\pi(a)(p+N):= ap+N.$$
Since $N$ is a left ideal, this is well-defined. It is clear that $\pi$ is a representation. It also intertwines the involution:
\\
\[
\begin{split}
\langl \pi(a^*)(p+N) \mid q+N \rangl &= \langl a^*p+N \mid q+N \rangl = L(q^*a^*p), \\
\langl p+N \mid \pi(a)(q+N) \rangl & = \langl p+N \mid aq+N \rangl = L(q^*a^*p).
\end{split}
\]
Finally, $\pi$ maps into $\cB(\cH)$.  Assume $p\in\mathcal A$ is
such that $\langl p+N \mid p+N\rangl =L(p^*p)=1$, and let $a\in \mathcal A$.
By archimedeanity, there is $\eta\in\mathbb N$ with $\eta-a^*a\in\SOS_{\cA}$. Then we have
\[
\begin{split}
0 & \leq \langl \pi(a)(p+N)\mid \pi(a)(p+N)\rangl  \\ &= 
 \langl ap+N \mid ap+N\rangl  = L(p^*a^*ap) \leq
\eta L(p^*p)=\eta,
\end{split}
\]
whence $\|\pi(a)\|\leq\sqrt{\eta}$. Thus $\pi:\cA\to\cB(\cH)$ is a $*$-representation.

It remains to verify \Cref{eq:fX0}.
For $f\in\cLI$ we have
\[
\pi(f)\psi = \pi(f)(1+N) = f+N.
\]
     Since $L(\cLI)=\{0\}$ and 
     $\cLI^*  \cLI \subseteq \cLI,$
     we have $\cLI \subseteq N$. Thus $\pi(f)\psi=f+N=0$, as desired.
     \end{proof}

\begin{ex}
An appealing class of algebras for which this theorem applies are 
group algebras $\C[G]$. Indeed, for every group element $g\in\C[G]$ 
we have $1-g^*g=0\in\SOS$. Since the set of bounded elements
\[
H=\{f\in\C[G]\mid \exists\eta\in\mathbb N:\, \eta-f^*f\in\SOS\}
\]
is a $*$-subalgebra \cite{vidav} containing $G$, we must have $H=\C[G]$ and thus
$\SOS$ is archimedean in $\C[G]$.
\end{ex}

    \begin{cor}
     \label{cor:2sidedIarchimedean}
     Suppose $\cA$ is a $*$-algebra, where $\SOS_{\cA}$ is {archimedean}, and let
     $\ide\subseteq\cA$ be a $*$-closed two sided ideal.
Then the following are equivalent:
\ben[\rm(i)]
\item there is a $*$-representation 
     $\pi$ such that $\pi(\ide)=\{0\}$;
     \item[\rm(i)'] $\zh^{\rm re,\cA}(\ide)\neq\varnothing$;
     \item\label{it:aaa2}
     $     -1 \not\in  \SOS_\cA + \ide .$
     \een 
 \end{cor}

\begin{proof}
The proof is the same as that of \Cref{cor:2sidedI},
just that we use \Cref{thm:AB} instead of \Cref{thm:rnssCHMN}.
\end{proof}

\ssec{\NullSSs and Perfect Games}
\label{ssec:NullSS_and_Perfect}

Combining the \NullSS of \Cref{ssec:QuantumGamesNullSS} with \detsets defined in \Cref{ssec:Characterizing_perfect_games} immediately gives a characterization of games with perfect commuting operator strategies in terms of left ideals and sums of squares of the universal game algebra $\uGA$. 

\begin{thm}
\label{thm:Game_NullSS}
For a nonlocal game $\game$ \dby a set $\dset \subseteq \uGA$ the following are equivalent:
\begin{enumerate}[\rm(i)]
\item
$\game$ has a perfect commuting operator strategy;
\item $-1 \notin \LIg{\dset} + \LIg{\dset}^* + \SOS_\uGA $.
\end{enumerate}
\end{thm}

\begin{proof}
Immediate from \Cref{thm:perfect_games_direc_zero,defn:dset}.
\end{proof}

An immediate corollary of \Cref{thm:Game_NullSS} comes from recalling the notions of determining sets $\validpoly, \invalidpoly$ defined in \Cref{sssec:valid_invaid_response_det_sets}. 

\begin{cor}
\label{cor:valOnenss}
Let $\game$ be a a nonlocal game, and $\validpoly, \invalidpoly$ be the determining sets associated with the game. Then the following are equivalent:
\begin{enumerate}[\rm(i)]
    \item $\covalue(\game) = 1$;
    \item $-1 \notin \LIg{\validpoly} + \LIg{\validpoly}^* + \SOS_{\uGA}$; 
    \item $-1 \notin \LIg{\invalidpoly} + \LIg{\invalidpoly}^* + \SOS_{\uGA}$.
\end{enumerate}
\end{cor}

\begin{proof}
Immediate from \Cref{thm:Game_NullSS,thm:perfect_games_direc_zero}.
\end{proof}

This corollary applies to all games (according to the definitions here) and characterizes which games do vs.~do not have a quantum strategy.
Unfortunately, the freedom given by the SOS terms
in this algebraic certificate can make this theorem  hard to use.
Hence, we turn next to situations with no SOS term.

\section{\NullSS without SOS and Subgroup Membership}
\label{sec:NoSOSNullSS}

It is helpful to divide the perfect game condition into two sub questions. The first is checking whether 
$
    -1 \in \cLI + \cLI^* 
$ that is, whether
$$
    1 \in \cLI + \cLI^*. 
$$
Intuitively, this question feels ``algebraic", and we will show in \Cref{ssec:Toric_Ideal_GA_Simplification} that in special cases it reduces to a subgroup membership problem. 

The second problem is checking whether 
\begin{align}
    -1 \in \cLI + \cLI^* + \SOS_\uGA 
\end{align}
given that 
\begin{align}
    1 \notin \cLI + \cLI^*.
\end{align}
This question is more analytic, and adds substantial complexity to applications.
In special cases we have that
\begin{align}
    1 \notin \cLI + \cLI^* \implies -1 \notin \cLI + \cLI^* + \SOS_\uGA
\end{align}
and hence the second problem is trivial. This seems closely related to the existence of projections which are conditional expectations and respect SOS. The next section investigates this link further.

\ssec{Conditional Expectations and SOS Projections}
\label{sssec:LeftIdealsNullSSSimplification}

We prepare to produce   simpler  a \NullSS with no SOS terms. This uses   existence of  SOS projections and  conditional expectations, notions we now study.

\sssec{Definitions}

Let $\cA$ be  $*$-algebra, and let
$\cC\subseteq\cA$. Recall
 $\SOS_\cC$ 
 denotes all 
sums of squares of members 
 of  $\cC$, i.e.,
 \[
\SOS_{\cC}=\Big\{ \sum_{i} c_i^*c_i \mid c_i\in\cC\Big\}.
 \]
If $\cC$ is a $*$-subalgebra, then $\SOS_{\cC}\subseteq\cC$.
 A  subtlety which is very important to us is: 
 while $f\in \SOS_\cA$ is in  $\cC$,
  $f$ may not be in 
  $ \SOS_\cC$.
 Similarly for $F \subseteq \cC $ we let
 $\cLI_\cC (F) \subseteq \cC$ denote
 the left ideal generated by $F$ in the algebra $\cC$. 
 When there is no confusion we may omit $\cC$ or $F$.

\begin{defn}

Given a unital $*$-algebra $\cA$, a (not necessarily unital) $*$-subalgebra $\cC$ and a projection $\cexp:\cA \rightarrow \cC$  (i.e., $\cexp^2 = \cexp$) onto $\cC$ satisfying 
    $\cexp(a)^* = \cexp(a^*)$ for all $a \in \cA$. Then
$\cexp$ is called a
\ben[\rm(1)]
\item 
\df{SOS Projection} if
$\cexp(\SOS_\cA) \subseteq \SOS_\cC$.

\item
\df{Conditional Expectation}
provided $\cC$ is unital and $\cexp$ satisfies
\begin{enumerate}[\rm (a)]
    \item\label{eq:bimodule} $\cexp(b_1 a b_2) = b_1 \cexp(a) b_2$ for all $a \in \cA$, $b_1, b_2 \in \cC$;
    \item $\cexp(1_\cA) = \cexp(1_\cB)$; 
        \item $\cexp(\SOS_\cA) \subseteq \SOS_\cA \cap\ \cC$.
\end{enumerate}
\item
\df{SOS Conditional Expectation} if $\cexp$
is a conditional expectation that also satisfies the SOS projection property  $\cexp(\SOS_\cA) \subseteq \SOS_\cC$. 
\een 

\end{defn}
Conditional expectations will typically be 
denoted $\cexp$.

\begin{rmk}
The bimodule property \ref{eq:bimodule} in the definition of a conditional expectation can be replaced by the seemingly weaker
one-sided version
\begin{enumerate}\setlength\itemindent{25pt}
 \item [(a')]
$\cexp(b a) = b \cexp(a)$ for all $a \in \cA$, $b \in \cC$.
\end{enumerate}
Indeed, given $a \in \cA$ and $c\in\cC$ we have
\[
\cexp(ac) =\cexp \big( (c^*a^*)^*\big) = 
\cexp \big( (c^*a^*)\big)^* = 
\big( c^* \cexp(a^*) \big) ^* =
\cexp(a^*)^* c = \cexp(a) c,
\]
as desired.
\end{rmk}

We now show existence of these mappings can simplify the nonlocal games \NullSS.

\sssec{\NullSS}

The next simple lemma explains the main use of  SOS and conditional expectation property.

\begin{lem}
\label{thm:cond_exp_SOS}
Given a unital $*$-algebra $\cA$,  a unital $*$-subalgebra $\cC$
and $\cLI$ a left ideal in $\cA$.
 If
an SOS conditional expectation $\cexp:\cA \to \cC$ exists, then
\ben[\rm(1)]
\item 
\label{it:P on LI}
$L:=\cexp(\cLI)$ is a left ideal in $\cC$;
\item 
 $-1 \notin \cLI + \cLI^* + \SOS_{\cA}$ \ \  \text{iff}  \ \ 
 $-1 \notin L + L^* + \SOS_{\cC}$.
\een 
\end{lem}

\begin{proof}
Suppose $b \in L$ and $c \in \cC$.
Then there is $\hat b \in \cLI$ 
satisfying $\cexp(\hat b) =b$.
These satisfy
$$
cb = c \cexp(\hat b) = \cexp(c \hat b) \in \cexp(\cLI)
= L
$$
establishing \Cref{it:P on LI}.

To prove the next item
assume  that 
$-1 \in \cLI + \cLI^* + \SOS_{\cA}$. 
Then 
\begin{align}
 -1=   \cexp(-1) &\in \cexp(\cLI + \cLI^* + \SOS_{\cA}) = L + L^* + \SOS_{\cC}.
\end{align}
The converse is obvious. 
\end{proof}

\begin{cor}
\label{cor:SOS_CE_LI}
Let $\algebra$ be a $*$-algebra, $\ssalg$ be a $*$-subalgebra of $\algebra$ and $\polyn\subseteq\ssalg$. 
Then, if a SOS conditional expectation $\cexp: \algebra \rightarrow \ssalg$ exists then 
\begin{align}
    -1 \notin \cLI(\polyn)_\algebra + \cLI(\polyn)_\algebra^* + \SOS_{\cA} \ \  \text{iff}  \ \ 
 -1 \notin \cLI(\polyn)_\ssalg + \cLI(\polyn)_\ssalg^* + \SOS_{\cC}
\end{align}
\end{cor}

\begin{proof}
We show that $\cexp (\cLI(\polyn)_\algebra) = \cLI(\polyn)_\ssalg$. Then the result is immediate from \Cref{thm:cond_exp_SOS}.

First, 
$\cexp$ is the identity map on $\cC$, whence
$\cLI(\polyn)_\ssalg \subseteq \cLI(\polyn)_\algebra$ implies
$\cLI(\polyn)_\ssalg \subseteq\cexp (\cLI(\polyn)_\algebra) $.
To show the other inclusion note any element $p \in \cLI(\polyn)_\algebra$ can be written as 
$
    p = \sum_\pol a_\pol \pol, 
$
with all $\pol \in \polyn$ and $a_\pol \in \algebra$. Then (using the bimodule property of $\cexp$ 
in the second equality) we see
\begin{align}
    \cexp(p) = \sum_\pol \cexp(a_\pol \pol) = \cexp(a_\pol) \pol \in \cLI(\polyn)_\ssalg, 
\end{align}
hence $\cexp (\cLI(\polyn)_\algebra) \subseteq \cLI(\polyn)_\ssalg$. Then $\cexp (\cLI(\polyn)_\algebra) = \cLI(\polyn)_\ssalg$ and the proof is complete. 
\end{proof}

Now we  prepare for finding 
a subalgebra $\cC$ that makes \Cref{thm:cond_exp_SOS} valuable.

\begin{lem} \label{lem:small_subalgebra_SOS}
Let $\algebra$ be a 
unital $*$-algebra and $\polyn$ be a set of elements in $\algebra$.
Also let $\cC$ denote the $*$-subalgebra 
of $\algebra$ generated by $\{\polyn, 1\}$ and $\cLI_\cC$ be the left ideal in $\cC$ generated by $\polyn$. Finally, let $\cC'$ be the subalgebra of $\algebra$ generated by $\polyn$. Then the following hold.
\begin{enumerate}[\rm(1)]

\item If $F=F^*$, then
$\cC'= \lide_C$.
\item
If $  1 $ is not in the (non-unital) $*$-subalgebra $\cC'$ generated by $F$,
then 
\begin{align}
    -1 \notin \cLI_{\cC} + \cLI_{\cC}^* + \SOS_{\cC}.
\end{align}
\end{enumerate}
\end{lem}

\begin{proof}
(1) By definition,
$\lide_C\supseteq F$ and $\lide_C$ is closed under addition and multiplication. Since $F=F^*$, $\lide_C$ is also closed under the
involution. It is thus contained in the $*$-algebra $\cC'$.
Conversely, 
since $F=F^*$, $\cC'$ is the algebra generated by $F$
and thus by definition contained in $\lide_C$.

(2) First note that any polynomial $p \in \cC$ can be written as 
\begin{align}
    p = p' + \alpha
\end{align}
where $\alpha \in \mathbb{C}$ and $p'\in\cC'$. It follows that any polynomial $q \in \cLI_\cC$ can be written as  a sum of terms of the form
\begin{align}
    q = \left(p' + \alpha\right) f 
\end{align}
with $f \in \polyn$. A similar description holds for any polynomial in $\cLI_\cC ^*$. Additionally, any polynomial $p'' \in \SOS_\cC$ can be written as 
\begin{subequations}
\begin{align}
    p'' &= \sum_i (p_{ i} + \alpha_i)^*(p_{ i} + \alpha_i) \\
    &= \sum_i \left( p_{i}^*p_i + \alpha_i^* p_{i} + \alpha_i p_{i}^* \right)  + \sum_i |\alpha_i|^2 \\
    &= \tilde p + \alpha''
\end{align}
\end{subequations}
with each $p_{i} \in \cC'$; hence $\tilde p\in\cC'$ and $\alpha''\in\R_{\geq0}$. 

Now assume for contradiction that $-1 \in \cLI_\cC + \cLI_\cC^* + \SOS_\cC$. 
Then, combining the above observations we can write 
\begin{align}\label{eq:-1one}
    - 1 = p' + \tilde p + \alpha ''
\end{align} 
with $p', \tilde p \in\cC'$ and $\alpha \in\R_{\geq0}$. Rearranging \Cref{eq:-1one} gives 
\begin{align}
    -(1 + \alpha'') &= p' + \tilde p \in\cC',
\end{align}
implying that $1\in\cC'$ since $1+\alpha''\in\R_{>0}$.
\end{proof}

Combining \Cref{lem:small_subalgebra_SOS,thm:cond_exp_SOS} results in the following ``SOS free'' \NullSS. 

\begin{thm} \label{thm:SOS_cond_exp_nullSS}
Let $\algebra$ be a 
unital $*$-algebra with archimedean $\SOS_\algebra$,
and let $\polyn=\polyn^*\subseteq\algebra$. Also let $\cC$ be the $*$-subalgebra of $\algebra$ generated by $\{\polyn, 1\}$.
If there exists an SOS conditional expectation $\algebra\to\cC$, then the following are equivalent:
\begin{enumerate}[\rm(i)]
    \item There exists a $*$-representation $\pi : \algebra \rightarrow \cB(\cH) $ and vector $\psi \in \cH$ with $\pi(\pol)\psi = 0$ for all $\pol \in \polyn$;
    \item[\rm(i)']
    $\zd^{\rm re,\algebra}(\polyn)\neq\varnothing$;
    \item $-1 \notin \cLI(F)_\algebra + \cLI(F)^*_\algebra + \SOS_\algebra$; 
    \item $-1 \notin \cLI(F)_\ssalg + \cLI(F)^*_\ssalg + \SOS_\ssalg$;
    \item $1 \notin \cLI(F)_\algebra + \cLI(F)^*_\algebra$;
    \item $1 \notin \cLI(F)_\ssalg + \cLI(F)^*_\ssalg$;
    \item 
     $  1 $ is not in the (non-unital) $*$-subalgebra $\cC'$ generated by $F$;
\item $1\notin \cLI(F)_\algebra$;
\item $1 \notin \cLI(F)_\ssalg$.
\end{enumerate}
\end{thm}

\begin{proof}
Items (i), (i)' are equivalent by definition.
We have (i) $\Leftrightarrow$ (ii) by the real \NullSS  \Cref{thm:AB}, (ii) $\Rightarrow$ (iii) by set inclusion and (iii) $\Rightarrow$ (ii) by \Cref{cor:SOS_CE_LI}.

Next, (ii) $\Rightarrow$ (iv)  $\Rightarrow$ (v) again by set inclusion. The equivalence (vi) $\Leftrightarrow$ (viii) follows from 
$\cC'=\cLI(F)_\ssalg$, the implications (iv) $\Rightarrow$ (vii) $\Rightarrow$ (viii) 
are obvious, and  (vi) $\Rightarrow$ (iii) by \Cref{lem:small_subalgebra_SOS}.
\end{proof}

\def\tor{{W}}
  \def\bbT{\mathbb T}

\ssec{The Group Algebra Simplification}
\label{ssec:Toric_Ideal_GA_Simplification}

Obtaining SOS projections or conditional expectations is challenging. 
We now give a class of tractable situations and a \NullSS
appropriate for many games.
The setting is a  group algebra 
$\mathbb{C}[G]$ and $\polyn\subseteq\algebra$, later to be chosen 
a set of binomials.
Here $G$ is a discrete group. 
 
To an  element $p\in \C[G]$,
$$
p= \sum_j^\gamma s_j g_j
\quad 
$$
one  associates 
the set
$\supp(p)=\{ g_1, \dots,  g_j\}$ of all group elements 
 appearing in $p$, called the \df{support of $p$.}
Similarly, 
$\mon(p)=\{s_1 g_1, \dots, s_\gamma
g_\gamma\}$
is the set of monomials appearing in $p$.
These notions are naturally extended to subsets $F\subseteq\C[G]$, e.g.
$\supp(F)=\bigcup_{p\in F}\supp(p)$.

We need the next theorem and string of lemmas
to obtain our \NullSS aimed at a large class of games.
Our first observation is that 
there is a natural SOS conditional expectation mapping from $\C[G]$ onto 
a type of envelope of a given set $F$.

\begin{thm} \label{thm:cond_exp_for_toric_ideals}
Consider a group algebra $\mathbb{C}[G]$ and let $\polyn\subseteq\C[G]$.
Then there is an SOS conditional expectation $\cexp:\C[G]\to
\CC [\frak G( \supp(F))]$.
\end{thm}

\begin{proof}
Since $\frak G( \supp(F))$ is by definition a subgroup of $G$, this follows
 from
\cite{savchuk2013unbounded}.
For context, the map $P$ is defined by
\[
\sum_{g\in G} a_g g \mapsto \sum_{g\in \frak G( \supp(F))} a_g g.
\]
It is an SOS conditional expectation by routine calculation; see
\cite[Example 5, Proposition 4]{savchuk2013unbounded} for details.
\end{proof}

\sssec{Relating the Subalgebra and Subgroup Membership Problems}
\label{sssec:Subalgebra-Sugbroup_membership_problem}

In this section we restrict 
to sets $F$ of binomials 
in a group algebra $\C[G]$
of the form $f=r -1$
where each $r$ is a  monomial, $r = \beta g$ for some
 $g \in G$ and $\beta \in \CC$.
Our starting point is the following lemma, which gives a standard form for elements in $\Alg_{\C[G]}(F)$.

\begin{lem}
\label{lem:polyalg_standard_form1}
Consider a group algebra $\mathbb{C}[G]$ and let $\polyn$ be a set of elements in $\C[G]$ with each $\pol \in \polyn$ of the form $\pol = \monom - 1$ with $\monom$ a monomial
and let $H
$
denote the subgroup of the invertible elements $\C[G]^{-1}$ in $\C[G]$
generated by 
$\mon(F)$.
Then any  $p \in \Alg_{\C[G]}(F)$ can be written in the form 
\begin{align}\label{eq:target0}
    p = \sum_{u,v\in H} \beta_{u,v} (u - v)
\end{align}
with $\beta_{u,v} \in \mathbb{C}$. 
\end{lem} 

\begin{proof}
We induct on the length of products needed to express
$p \in \Alg_{\C[G]}(F)$. If $p$ is a linear combination of
elements of $F$, say
$
p=\sum_j \beta_j(r_j-1),
$
this is immediate since $r_j,1\in H$.

As the set of expressions of the form given in \Cref{eq:target0} 
is closed under linear combinations, 
for the induction step
it suffices to consider the product of
$u-v,u'-v'$ for $u,u',v,v'\in H$. 
But the product
\[
(u-v)(u'-v') = (uu'-uv') + (vv'-vu')
\]
is clearly of the form \Cref{eq:target} since
$uu',uv',vv',vu'\in H$, thus finishing the proof.
\end{proof}

\begin{lem}
\label{lem:polyalg_standard_form}
Consider a group algebra $\mathbb{C}[G]$ and let $\polyn$ be a set of elements in $\C[G]$ with each $\pol \in \polyn$ of the form $\pol = \monom - 1$ with $\monom$ a monomial
and let $\monomgp 
$
denote the subgroup of $\C[G]^{-1}$
generated by 
$\mon(F) \cup\mon(F^*)$
Then any element $p$ in the $*$-subalgebra generated by $F$, 
$p \in \Alg^*_{\C[G]}(F)$, can be written in the form 
\begin{align}\label{eq:target}
    p = \sum_{u,v\in\monomgp} \beta_{u,v} (u - v)
\end{align}
with $\beta_{u,v} \in \mathbb{C}$. 
\end{lem} 

\begin{proof}
Immediate from  \Cref{lem:polyalg_standard_form1}.
\end{proof}

\begin{lem} \label{lem:subalgebra_membership_subgroup_membership}
Assume the hypotheses of \Cref{lem:polyalg_standard_form}
are in force.
Then 
\begin{align}
    1 \in \Alg^*_{\C[G]}(F) \quad \Leftrightarrow \quad \monomgp \cap \mathbb{C} \varsupsetneq \{1\}.
\end{align}
\end{lem}

\begin{proof}
First we note that the result is trivial if there exists an element $\pol = \monom - 1 \in \polyn$ with $\monom^* \monom \neq 1$. 
That is, $\monom=\beta g$ with $|\beta|\neq1$.
 Namely,
\begin{align}
    (r-1)^*(r-1) + (r-1) + (r-1)^* = r^*r - 1 = \beta^* \beta - 1 \in \Alg^*_{\C[G]}(F)
\end{align}
and since $\beta^* \beta - 1 \neq 0$ we can conclude $ 1 \in\Alg^*_{\C[G]}(F)$. At the same time $r r^* = \beta \beta^* \in \Delta \cap \mathbb{C}$ and $\beta \beta^* \neq 1$. This proves the result in this special case. 
In what follows we can assume $\monom \monom^* = 1$ for all $\monom - 1 \in \polyn$.

$(\Leftarrow)$ First note that for all monomials $\monom$ and $\monom'$ we have 
\begin{align}
    (\monom - 1)(\monom' - 1) + (\monom - 1) + (\monom' - 1)  = \monom\monom' - 1.
\end{align}
Similarly, 
\[
(r-1)^* = r^* - 1 = r^{-1} - 1.
\]
It follows that for any $r \in \monomgp$ we have $r - 1 \in \Alg^*_{\C[G]}(F)$. Then, if $1\neq\beta\in\Delta\cap\C$, we have 
$\beta - 1 \in\Alg^*_{\C[G]}(F)$, whence $1\in\Alg^*_{\C[G]}(F)$.

$(\Rightarrow)$ 
To prove the result in the other direction we assume for contradiction that $\beta  \notin \monomgp$ for all $\beta \in \mathbb{C}\setminus\{1\}$ and that $1 \in \Alg^*_{\C[G]}(F)$. 
Then, using \Cref{lem:polyalg_standard_form}, we can write 
\begin{subequations}
\begin{align}
    1 &= \sum_{u,v\in\Delta} \beta_{u,v} (u - v)  
=\sum_{u,v} \beta_{u,v} u - \sum_{u,v} \beta_{u,v}v 
=\sum_{u,v} \beta_{u,v} u - \sum_{u,v} \beta_{v,u}u \\
    & = \sum_{u,v} u \left(\beta_{u,v} - \beta_{v,u}\right)
   = \sum_{u} u \sum_{v} \left(\beta_{u,v} - \beta_{v,u}\right)\label{eq:last}
\end{align}
\end{subequations}
where we relabeled $u$ and $v$ in the last term in the sum on the first line. 
By assumption, $\beta  \notin \monomgp$ for all $\beta \in \mathbb{C}\setminus\{1\}$. Thus the terms 
$u \sum_{v} \left(\beta_{u,v} - \beta_{v,u}\right)$ in the last equality of \Cref{eq:last}
are linearly independent since the underlying group elements of distinct $u\in\monomgp$ are distinct.
We conclude
\begin{subequations}
\begin{align}
    \sum_v (\beta_{1,v} - \beta_{v,1}) &= 1,\\
    \sum_v (\beta_{u,v} - \beta_{v,u}) &= 0
\end{align}
\end{subequations}
for all $u \neq 1$. But this is a contradiction, since 
\[
    \sum_v (\beta_{1,v} - \beta_{v,1}) + \sum_{u \neq 1} \sum_v (\beta_{u,v} - \beta_{v,u}) = \sum_{u,v} (\beta_{u,v} - \beta_{v,u}) = 0 . \qedhere
\]
\end{proof}

\sssec{NC  Left \NullSS without SOS Terms}

\def\bbT{{\mathbb T}}

Now we combine the results in \Cref{sssec:LeftIdealsNullSSSimplification,sssec:Subalgebra-Sugbroup_membership_problem} to obtain the following specialized \NullSS.

\begin{thm} \label{thm:noSOSNullSS}
Consider a group algebra $\mathbb{C}[G]$ and let $\polyn$ be a set of elements in $\C[G]$ with each $\pol \in \polyn$ of the form $\pol = \monom - 1$ with $\monom$ a monomial
and let $\monomgp 
$
denote the subgroup of $\C[G]^{-1}$
generated by 
$\mon(F\cup F^*)$.

Then the following are equivalent:
\begin{enumerate}[\rm(i)]
    \item There exists a $*$-representation $\pi : \C[G] \rightarrow \cB(\cH) $ and vector $\psi \in \cH$ with $\pi(\pol)\psi = 0$ for all $\pol \in \polyn$; 
    \item [\rm(i)']
    $\zd^{\rm re,\C[G]}(\polyn)\neq\varnothing$;
    \item $-1 \notin \cLI(F)_{\C[G]} + \cLI(F)^*_{\C[G]} + \SOS_{\C[G]}$;
    \item $1 \notin \cLI(F)_{\C[G]} + \cLI(F)^*_{\C[G]}$;
    \item $1 \notin \cLI(F)_{\Alg^*_{\C[G]}(F)} + \cLI(F)^*_{\Alg^*_{\C[G]}(F)}$;
    \item $1 \notin \Alg^*_{\C[G]}(F)$;
    \item $\monomgp \cap \mathbb{C} = \{1\}$. 
\end{enumerate}
Moreover, if 
each $\pol \in \polyn$ is of the form $\pol = \beta w - 1$ with $w\in G$ and $\beta\in\C$ with $|\beta|=1$, then these statements are also equivalent to
\begin{enumerate}[\rm(i)]
 \setcounter{enumi}{6}
      \item $1 \notin \cLI(F)_{\C[G]}$;
      \item $1\not\in \cLI(F)_{\Alg^*_{\C[G]}(F)} $.
\end{enumerate}
\end{thm}

\begin{proof}
We have 
(i) $\Leftrightarrow$ (i)' by definition, 
(i) $\Leftrightarrow$ (ii) by the real \NullSS  \Cref{thm:AB}, and (ii) $\Rightarrow$ (iii) $\Rightarrow$ (iv)
by set inclusion.

To show (iv) $\Rightarrow$ (v), assume $1 \in \Alg^*_{\C[G]}(F)$,
i.e.,
\beq\label{eq:salvage1}
1= p_1 f_1 + p_2 f_2^* 
\eeq
for some $f_j\in F$ and $p_i\in\Alg^*_{\C[G]}(F\cup\{1\}) $.
Write $f_2=\beta_2 w_2-1$ for $\beta_2\in\C$ and $w_2\in G$. Then
$f_2^*=\beta_2^* w_2^*-1$, and
\beq\label{eq:salvage2}
f_2^*f_2+f_2+f_2^* = |\beta_2|^2-1.
\eeq
If $|\beta_2|^2\neq1$, then \Cref{eq:salvage2} immediately implies
\[
1 \in \cLI(F)_{\Alg^*_{\C[G]}(F)} + \cLI(F)^*_{\Alg^*_{\C[G]}(F)}.
\]
If $|\beta_2|^2=1$, then from \Cref{eq:salvage2} we deduce
\beq\label{eq:salvage3}
f_2^* = -f_2^*f_2-f_2 \in \cLI(F)_{\Alg^*_{\C[G]}(F)},
\eeq
whence \Cref{eq:salvage1} yields
\[
1= p_1 f_1 - p_2 f_2^* f_2- p_2 f_2 \in  \cLI(F)_{\Alg^*_{\C[G]}(F)}
\subseteq \cLI(F)_{\Alg^*_{\C[G]}(F)} + \cLI(F)^*_{\Alg^*_{\C[G]}(F)},
\]
as desired.

Next, items (v) and (vi) are equivalent by \Cref{lem:subalgebra_membership_subgroup_membership}, and
(v) $\Rightarrow$ (ii) by \Cref{lem:small_subalgebra_SOS} 
and \Cref{cor:SOS_CE_LI}.
Here we use
that an SOS conditional expectation mapping $\C[G] \to \Alg^*_{\C[G]}(F\cup\{1\})=\CC [\frak G( \supp(F))]$ exists by \Cref{thm:cond_exp_for_toric_ideals}. 

Finally, we tackle the moreover statement.
Implications (iii) $\Rightarrow$ (vii) $\Rightarrow$ (viii) are obvious.
To conclude we prove (viii) $\Rightarrow$ (iv). 
Thus assume $1 \in \cLI(F)_{\Alg^*_{\C[G]}(F)} + \cLI(F)^*_{\Alg^*_{\C[G]}(F)}$, i.e.,
\beq\label{eq:2salvage1}
1= \sum p_i f_i + \sum g_j^* q_j
\eeq
for some $p_i,q_j\in\Alg^*_{\C[G]}(F\cup\{1\})$ and $f_i,g_j\in F$.
Write $q_j=q_{j1} f_{j1} + q_{j2} g_{j2}^*$ for some
$q_{ji}\in\Alg^*_{\C[G]}(F\cup\{1\})$ and $f_{j1},g_{j2}\in F$.
As in \Cref{eq:salvage3} we write
\beq\label{eq:2salvage2}
g_{j2}^* = -g_{j2}^*g_{j2}-g_{j2} \in \cLI(F)_{\Alg^*_{\C[G]}(F)}.
\eeq
Using \Cref{eq:2salvage2} in \Cref{eq:2salvage1}
leads to
\[
\begin{split}
1 & =\sum p_i f_i + \sum g_j^* q_j=
\sum p_i f_i + \sum g_j^* (q_{j1} f_{j1} + q_{j2} g_{j2}^*) \\
& = 
\sum p_i f_i 
+ \sum g_j^* q_{j1} f_{j1} 
- \sum g_j^* q_{j2}g_{j2}^*g_{j2}
- \sum g_j^* q_{j2}g_{j2} \\
& \in 
\cLI(F)_{\Alg^*_{\C[G]}(F)},
\end{split}
\]
as desired.
\end{proof}

\ssec{\NullSSs for Perfect \TGames}
\label{sec:unitaryGames}

We now apply the simplified \NullSS of this section to nonlocal games. We first recall the definition of \tgames introduced in \Cref{sssec:tgames}. 

\begin{defn}[repeated]
A game $\game$ is called a \tgame if there exists a group $G$ with $\uGA \cong \mathbb{C}[G]$ and $\game$ is determined by a set of elements 
\begin{align}
    \dset = \{\beta_i g_i - 1\}
\end{align}
with each $\beta_i \in \mathbb{C}$ and $g_i \in G$. In this case we say $\game$ is \tdby the set $\dset$ and call the elements $\beta_i g_i$ clauses of $\dset$. 
\end{defn}

We let the set $\clauseset$ denote all the clauses of $\dset$. Now the  
following characterization of \ugames with perfect commuting operator strategies
is a quick consequence of our \NullSS 
\Cref{thm:noSOSNullSS}.

\begin{thm} \label{thm:perfect_unitary_games}
Let $\game$ be a game which is \tdby a set $\dset = \clauseset - 1$, and let $\clausegp$ be the group of elements in $\uGA$ generated by $\clauseset \cup \clauseset^*$. Then $\game$ has a perfect commuting operator strategy iff the following equivalent criteria are satisfied:
\begin{enumerate}[\rm(i)]
    \item $1 \notin 
    \LIg{\dset}+\LIg{\dset}^*$;
    \item $\clausegp\cap\CC=\{1\}$.
\end{enumerate}
Moreover, if $|\beta|=1$ for each
each $h=\beta g  \in\mathscr H$ then these statements are also equivalent to
\begin{enumerate}[\rm(i)]
 \setcounter{enumi}{2}
    \item $1  \notin 
    \LIg{\dset}$.
\end{enumerate}

\end{thm}

\begin{proof}
By definition, $\game$ has a perfect commuting operator strategy iff there exists a $*$-representation $\pi$ mapping $\uGA$ into bounded operators on a Hilbert space $\cH$ and a state $\state$ in $\cH$ with $\pi(\clause - 1) \state = 0$ for all clauses $\clause\in\clauseset$. But this is equivalent to the statement $\zd^{\rm re,\C[G]}(\dset)\neq\varnothing$. Then the result follows directly from  \Cref{thm:noSOSNullSS}. 
\end{proof}

\sssec{\TGames and the Subgroup Membership Problem}
\label{sssec:tgames_smp}

Condition $\rm{(ii)}$ of \Cref{thm:perfect_unitary_games} relates existence of perfect commuting operator strategies to a question of membership of certain elements in a group. We now translate this question to a standard instance of the subgroup membership problem. 

First, observe that if any clause $\beta_i g_i \in \clauseset$ has $\abs{\beta_i} \neq 1$ then the game $\game$ trivially cannot have a perfect commuting operator strategy, since 
\begin{align}
    \beta g \left( \beta g \right)^{*} = \abs{\beta}^2 \neq 1 \in \clausegp.
\end{align}
Then we restrict our attention to the case where $\abs{\beta_i} = 1$ for all $\beta_i g_i \in \clauseset$. Let $\betagp$ be the group generated by all the $\beta_i$ under multiplication, and note $\betagp$ is an abelian group consisting of a subset of the unit circle. 
Then the statement $\clausegp\cap\CC=\{1\}$ is equivalent to 
$\clausegp\cap B =\{1\}$,
i.e., $\beta \not\in H$ for any $1\neq\beta \in \betagp$.
 This is exactly equivalent to one (or several) instances of the subgroup membership problem. 

For many games the group $\betagp$ is very simple, containing just a few elements. In this case, alternate notation can be used for elements $\beta \in \betagp$ (for example, the $J$ element used in \cite{cleve2017perfect} or the $\sigma$ element in \cite{watts20203xor}). We see an example of this next.

\sssec{Mod \texorpdfstring{$r$}{r} Games}

\label{ex:mod_m_games_perfect_condition}

We now apply the machinery described previously in this section to Mod $r$ Games. From \Cref{sssec:Mod $r$ Games} we have that Mod $r$ games are \tdby a set of elements
\begin{align}
    \left\{(\exp(-2\pi i /r))^{s_t} \prod_{\alpha \in [k]} \left(\cycla{j_{t}(\alpha)}{\alpha}(\alpha)\right)^{d_t(\alpha)} - 1\right\}_{t \in [T]}
\end{align}
with the $\cycla{j_{t}(\alpha)}{\alpha}$ cyclic unitaries. 

Following the notation introduced in \Cref{sssec:tgames_smp} we define
$\betagp$ to be the group generated by the set of elements $\{(\exp(-2\pi i /r))^{s_t}\}_{t \in [T]}$. We note that $\betagp$ is isomorphic to a subgroup of $\mathbb{Z}_r$ and for notational convenience introduce the element $\zeta$ as shorthand for $\exp(-2\pi i /r)$. Then we can write clauses $\clause_t \in \clauseset$ as 
\begin{align}
    h_t = \zeta^{s_t} \prod_{\alpha \in [k]} \cycla{j_{t}(\alpha)}{\alpha}.
\end{align}
It follows that a Mod $m$ game has a perfect commuting operator strategy iff $\zeta^s \notin \clausegp$ for all $s \in \{1,2,\ldots,r-1\}$. 

We note that if $r$ is prime there is an immediate simplification of this characterization, since $\zeta^s \in H \Leftrightarrow \zeta^{rs+1} = \zeta \in H$ and hence a game has a perfect commuting operator strategy iff $\zeta \notin \clausegp$. We also note that the Mod 2 game (i.e., XOR game) version of this condition is equivalent to Theorem 2.1 in \cite{watts20203xor}.

\sec{Gr\"obner Basis Algorithm Tailored to Games}
\label{sec:GBAlgorithmFull}

\def\cR{{\mathcal R}}

What arises in this paper are sums of left and two sided ideals. We present here a kludge for using a conventional two sided ideal nc Gr\"obner Basis (GB) algorithm to solve these mixed ideal
problems. While a special purpose algorithm could be more efficient, the procedure here can be very convenient in practice. It has certainly been valuable to the authors.
We also alert the reader to \Cref{sec:GBExamples1} which applies Gr\"obner Bases plus a \NullSS to synchronous games.

At the core of our algorithm is the following observation:

\begin{prop}\label{prop:preGB}
Consider an ideal $\ide$ and left ideal
$\lide$ in a free algebra $\fa$.
For $f\in\fa$, we have
\[
f\in \ide+\lide\quad \iff\quad f\xi \in (\ide+\lide \xi)_{\fay},
\]
where $(\ide+\lide \xi)_{\fay}$ denotes the two-sided
ideal of $\fay$ generated by $\ide+\lide \xi$.
In particular,   \[
1\in \ide+\lide\quad \iff\quad\xi \in (\ide+\lide \xi)_{\fay}.\]
\end{prop}

\begin{proof}
The forward implication is obvious, so assume
$f\xi \in (\ide+\lide \xi)_{\fay},$ i.e.,
\beq\label{eq:fy1}
f\xi=\sum_j f_j a_j g_j + \sum_i p_i s_i\xi q_i
\eeq
for some $f_j,g_j,p_i,q_i\in\fay$, $a_j\in\ide$ and $s_i\in\lide$.

Each element $r$ in $\fay$ can be written uniquely as
$r=r_0+r_1$, where $r_0\in\fa$ and $r_1\in (\xi)_{\fay}$.
We keep
only terms of degree $\geq1$ in $\xi$ in  \Cref{eq:fy1} to obtain
\beq\label{eq:fy2}
f\xi=\sum_j f_{j0} a_j g_{j1} 
+
\sum_j f_{j1} a_j g_{j0} +
 \sum_i p_{i0} s_i\xi q_{i0}.
\eeq
Now extract all terms that end in $\xi$ to get
\beq\label{eq:fy3}
f\xi=\sum_j f_{j0} a_j \tilde g_{j1} \xi 
+
 \sum_i p_{i0} \hat q_{i0} s_i \xi
\eeq
for some $\tilde g_{j1}\in\fa$ and $\hat q_{i0}\in\mathbb k$.
Cancelling $\xi$ on the right in \Cref{eq:fy3} leads to
$f\in\ide+\lide$.
\end{proof}

\subsection{Gr\"obner Basis Algorithm}
\label{sec:GBAlgorithm}

\paragraph{Algorithm}
Suppose $a_1,\ldots,a_m,b_1,\ldots,b_n\in\fa$ are given,
and let $\ide=(a_1,\ldots,a_m)$, $\lide=\fa b_1+\cdots +\fa b_n$ be a two-sided and left ideal in $\fa$, respectively. 
Add a new variable $\xi$,
to form the two-sided ideal
$(a_1,\ldots,a_m,b_1\xi,\ldots,b_n\xi)_{\fay}$
and compute its Gr\"obner basis $B$, with respect to any admissible monomial order.
Then
 $$ 1\in \ide+\lide \qquad  \text{iff} \qquad 
1\in B \ \  \text{or} \ \  
 \xi\in B.$$
 Here $1\in B$ iff $1\in\ide$.
 \qed
 
 \paragraph{Basics of NC GB}
 
We assume the reader has a familiarity  with Gr\"obner Bases; a standard reference in the commutative setting
is \cite{cox15book}, 
while \cite{moraNCGB,greenNCGB} describe the appropriate nc analogs.
NC GBs have properties very similar  to those of for traditional commutative GB with the dramatic exception that an NC GB might not be finite. 

Fixing a monomial order, a subset $B$ of an ideal
$\ide \subseteq\fa$ is a GB if the set
of leading terms $\LT(B)$ of elements of $B$ generates
the same ideal $\LT(\ide)$ as all leading terms of $\ide$. Roughly speaking, the nc Bucherger criterion \cite[Theorem 5.1]{moraNCGB}
states that $B$ is a GB iff each $S$-polynomial
built off $B$
can be expressed in terms of the elements of $B$
with lower degrees. 
Thus
algorithms for building NC GBs work by producing 
$S$-polynomials for pairs
$a_i,a_j$ of (not necessarily distinct)  polynomials which are generators  of $\ide$.
S-polynomials are ones  of the form
\beq\label{eq:Spoly}
S_{i,j}(w_i,w_i',w_j,w_j')=
\frac1{\LC(a_i)}w_ia_iw_i'-
\frac1{\LC(a_j)}w_ja_jw_j'
\eeq
for words $w_i,w_i',w_j,w_j'$ in $x$ satisfying
\beq\label{eq:obstr}
w_i\LT(a_i)w_i'=
w_j\LT(a_j)w_j'.
\eeq
Here $\LC(a)$ denotes the coefficient of the leading term of $a$.

At each step 
in construction of a GB, one considers a collection of polynomials $\hat B$.
If the ``remainder'' of an $S$-polynomial
after division by $\hat B$ is nonzero, one adds it
to $\hat B$ and repeats the process.
Ultimately (maybe in an infinite number of steps)  $\hat B$
grows to a GB $B$.

Left ideals have what we call left GB.
Algorithms for producing left GBs
are very similar to 
algorithms for producing GBs for 2 sided ideals.
Now common multiples must be made by 
multiplying on the
left (not the right),
and the key matches are 
\beq 
\label{eq:matches}
w_i\LT(a_i)=
w_j\LT(a_j)
\eeq 
for words $w_i,w_j$ in $x$.

We refer the reader to \cite{MR98,phdNCGB}
for  excellent references on this.

\paragraph{Justification}
  
 Now behold a few properties of our GB.
 \ben[\rm (1)]
 \item 
 The set of polynomials in 
 $B$ which do not contain $\xi$ is itself a 
 Gr\"obner basis for $\ide$.
 
  \item 
 The set of polynomials in 
 $B$ which do  contain $\xi$, after  $\xi$ is removed,
 is itself a 
 Left Gr\"obner basis for $\lide$.
 Moreover, the left $S$-polynomials which occur
 in the left GB algorithm for $\lide$
 are the ``same'' as those which occur in the
 standard two sided ideal GB algorithm for
 $\lide \xi$; thus their  run times  are similar.

 \item
 \label{it:ugiGB}
 The ``natural generators''
 $x_i^2-1$, $x_i y_j -y_j x_i $,
 etc., are themselves a GB for the universal game ideal,  $\uGI$.
This is easily checked, say with a computer, by focusing on the subset of all
polynomials involving only
two variables $x_i,y_j$.
 
 \een

\def\faab{\mathbb k\langle a_0,a_1,b_0,b_1\rangle}
\def\faaby{\mathbb k\langle a_0,a_1,b_0,b_1,\xi\rangle}
\def\faxy{\mathbb k\langle x_0,x_1,y_0,y_1\rangle}
\def\faxyy{\mathbb k\langle x_0,x_1,y_0,y_1,\xi\rangle}

\ssec{Examples of GBs}

In the forthcoming examples we  use the natural graded lexicographic order defined by
$$ x_i <y_j < z_k<\xi \qquad \text{and} \quad  x_0 <x_1 < \cdots , \quad \text{etc.} 
$$

\begin{ex}
Consider the CHSH game. 
It has a 
perfect commuting operator strategy
iff there exist signature 
matrices $X_0,X_1,Y_0,Y_1$ with $X$s commuting with $Y$s so that for some state $\psi$ we have
\beq\label{eq:CHSH}
\begin{split}
X_0Y_0\psi&=\psi\\
X_0Y_1\psi&=\psi\\
X_1Y_0\psi&=\psi\\
X_1Y_1\psi&=-\psi .
\end{split}
\eeq

The universal game ideal $\uGI$ for 
 a 2-player, 2-question game, such as CHSH,  with answers denoted $x_0, x_1, y_0,y_1$
 lies in the free algebra 
$\faxy$
and is generated by
the 4 signature properties $x_i^2- 1, y_j^2- 1 $ and 4 commuting relations $x_iy_j- y_jx_i$.
The GB for $\uGI$
is itself, by \Cref{it:ugiGB}.

So take $\ide=\uGI $
and take 
$\lide$ generated by $x_0y_0-1,\ x_0 y_1-1,\ x_1y_0-1,\ x_1y_1+1$.
This encodes the CHSH game.
The Gr\"obner basis $B$ for $(\uGI +\lide \xi)_{\faxyy}$ is given by the defining relations for
 $\uGI$ and $\xi$.
Since $\xi\in B$, by \Cref{prop:preGB} we
have $1\in\uGI +\lide$, whence \Cref{eq:CHSH}
does not have a solution.
\end{ex}

\ssec{Construction of solutions to the game equations}

Next we show how to attempt finding solutions to the game equations once we have a GB. We do this with an example.

\begin{ex}[GHZ 3 player game]\label{ex:GHZ3}

We seek
symmetries $x_i,y_i,z_i$, $0\leq i\leq1$,
with $x$'s commuting with $y$'s and $z$'s, and
$y$'s commuting with $z$'s, so that the four operators
\beq \label{eq:GHZ}
\begin{aligned}
-1 - x_0  y_0  z_1,\ -1 - x_0  y_1  z_0,\ -1 - x_1  y_0  z_0,  \
1-x_1y_1z_1
\end{aligned}
\eeq
have a common nonzero kernel vector.

In the free algebra $\faxyz$ consider
its ideal
$\ide$ encoding the signature relations $x_i^2=x_i, y_j^2=y_j, z_j^2=z_j$ and commuting relations $x_iy_j=y_jx_i,
x_iz_j=z_jx_i,z_iy_j=y_jz_i$,
 and its
left ideal $\lide$ generated by
the polynomials in \Cref{eq:GHZ}.
 The Gr\"obner basis $B$ for $(\ide+\lide \xi)_{\faxyzpsi}$ is given by: \\
 \mbox{} \qquad \qquad 
 the signature relations, commuting relations, and the following 18 polynomials (all of whom are divisible by $\xi$ on the right)
 \beq\label{eq:GB8}
 \small
 \begin{aligned}
y_0  \xi + x_0  z_1  \xi,\  z_1  \xi + x_0  y_0  \xi,\  
x_0  \xi + y_0  z_1  \xi,\  y_1  \xi + x_0  z_0  \xi,\  
z_0  \xi + x_0  y_1  \xi,\\  x_0  \xi + y_1  z_0  \xi,\  
y_0  \xi + x_1  z_0  \xi,\  z_0  \xi + x_1  y_0  \xi,\  
x_1  \xi + y_0  z_0  \xi,\  -y_1  \xi + x_1  z_1  \xi,\\  -z_1  \xi + 
 x_1  y_1  \xi,\  -x_1  \xi + y_1  z_1  \xi,\
x_0  x_1  \xi + y_1  y_0  \xi,\  -x_0  x_1  \xi + 
 z_0  z_1  \xi,\\  -x_1  x_0  \xi + z_1  z_0  \xi,\  -x_0  x_1  \xi + 
 y_0  y_1  \xi,\
x_0  x_1  \xi + x_1  x_0  \xi,\  -y_1  z_0  \xi + 
 x_1  x_0  x_1  \xi.
\end{aligned}
\eeq

We shall use $B$ to construct a solution to our system. 
Let $$V_0=\faxyzpsi/(\ide+\lide \xi),$$ and let $f\mapsto \overline f$ denote the quotient map $\faxyzpsi\to V_0$. 
The subspace $$V:=\overline{\faxyz\xi}\subseteq V_0$$ is eight dimensional, spanned by $B=\{\xi ,\  x_0  \xi ,\  x_1  \xi ,\  y_0  \xi ,\  y_1  \xi ,\  z_0  \xi ,\  z_1  \xi ,\  x_0  x_1  \xi \}$.

Each of the nc variables $x_i,y_i,z_i$ acts on $V$ from the left,
and with respect to the basis $B$ we obtain matrix representations as follows:
\[
\scriptsize
\begin{split}
\widehat x_0 & =
\begin{pmatrix}
0 & 1 & 0 & 0 & 0 & 0 & 0 & 0 \\
 1 & 0 & 0 & 0 & 0 & 0 & 0 & 0 \\
 0 & 0 & 0 & 0 & 0 & 0 & 0 & 1 \\
 0 & 0 & 0 & 0 & 0 & 0 & -1 & 0 \\
 0 & 0 & 0 & 0 & 0 & -1 & 0 & 0 \\
 0 & 0 & 0 & 0 & -1 & 0 & 0 & 0 \\
 0 & 0 & 0 & -1 & 0 & 0 & 0 & 0 \\
 0 & 0 & 1 & 0 & 0 & 0 & 0 & 0
\end{pmatrix}, \quad
\widehat x_1  =
\begin{pmatrix}
0 & 0 & 1 & 0 & 0 & 0 & 0 & 0 \\
 0 & 0 & 0 & 0 & 0 & 0 & 0 & -1 \\
 1 & 0 & 0 & 0 & 0 & 0 & 0 & 0 \\
 0 & 0 & 0 & 0 & 0 & -1 & 0 & 0 \\
 0 & 0 & 0 & 0 & 0 & 0 & 1 & 0 \\
 0 & 0 & 0 & -1 & 0 & 0 & 0 & 0 \\
 0 & 0 & 0 & 0 & 1 & 0 & 0 & 0 \\
 0 & -1 & 0 & 0 & 0 & 0 & 0 & 0 
\end{pmatrix}, \\
\widehat y_0 & =
\begin{pmatrix}
 0 & 0 & 0 & 1 & 0 & 0 & 0 & 0 \\
 0 & 0 & 0 & 0 & 0 & 0 & -1 & 0 \\
 0 & 0 & 0 & 0 & 0 & -1 & 0 & 0 \\
 1 & 0 & 0 & 0 & 0 & 0 & 0 & 0 \\
 0 & 0 & 0 & 0 & 0 & 0 & 0 & 1 \\
 0 & 0 & -1 & 0 & 0 & 0 & 0 & 0 \\
 0 & -1 & 0 & 0 & 0 & 0 & 0 & 0 \\
 0 & 0 & 0 & 0 & 1 & 0 & 0 & 0
\end{pmatrix}, \quad
\widehat y_1  =
\begin{pmatrix}
 0 & 0 & 0 & 0 & 1 & 0 & 0 & 0 \\
 0 & 0 & 0 & 0 & 0 & -1 & 0 & 0 \\
 0 & 0 & 0 & 0 & 0 & 0 & 1 & 0 \\
 0 & 0 & 0 & 0 & 0 & 0 & 0 & -1 \\
 1 & 0 & 0 & 0 & 0 & 0 & 0 & 0 \\
 0 & -1 & 0 & 0 & 0 & 0 & 0 & 0 \\
 0 & 0 & 1 & 0 & 0 & 0 & 0 & 0 \\
 0 & 0 & 0 & -1 & 0 & 0 & 0 & 0 
\end{pmatrix}, \\
\widehat z_0 & =
\begin{pmatrix}
 0 & 0 & 0 & 0 & 0 & 1 & 0 & 0 \\
 0 & 0 & 0 & 0 & -1 & 0 & 0 & 0 \\
 0 & 0 & 0 & -1 & 0 & 0 & 0 & 0 \\
 0 & 0 & -1 & 0 & 0 & 0 & 0 & 0 \\
 0 & -1 & 0 & 0 & 0 & 0 & 0 & 0 \\
 1 & 0 & 0 & 0 & 0 & 0 & 0 & 0 \\
 0 & 0 & 0 & 0 & 0 & 0 & 0 & 1 \\
 0 & 0 & 0 & 0 & 0 & 0 & 1 & 0 
\end{pmatrix}, \quad
\widehat z_1  =
\begin{pmatrix}
0 & 0 & 0 & 0 & 0 & 0 & 1 & 0 \\
 0 & 0 & 0 & -1 & 0 & 0 & 0 & 0 \\
 0 & 0 & 0 & 0 & 1 & 0 & 0 & 0 \\
 0 & -1 & 0 & 0 & 0 & 0 & 0 & 0 \\
 0 & 0 & 1 & 0 & 0 & 0 & 0 & 0 \\
 0 & 0 & 0 & 0 & 0 & 0 & 0 & -1 \\
 1 & 0 & 0 & 0 & 0 & 0 & 0 & 0 \\
 0 & 0 & 0 & 0 & 0 & -1 & 0 & 0 
\end{pmatrix}.
\end{split}
\] 
It is easy to verify that these solve the system \Cref{eq:GHZ} with common kernel vector $e_1$. 
\end{ex}

\begin{rmk}
For a 3-XOR game, 
if the
restricted GNS
construction (as in \Cref{ex:GHZ3})
 gives a finite-dimensional space of dim $> 8$, then the solution to 
the game may not be  unique.
Indeed,  MERP is one solution;
it has dim 8. Other solutions 
may
show up in the GNS construction
via it block diagonalizing and
the other solutions being 
blocks. These blocks might all be unitarily equivalent  to
the MERP solution; otherwise the game has multiple solutions.

\end{rmk} 

The problem of establishing uniqueness of solution to the game has  been studied  and there is a standard technique available. 
cf \cite{cui2020generalization}.
It can be tried on a game for which the bias and optimal value 
satisfy
$$
\omega^*- \Phi_\game = SOS
$$
exactly (as opposed to approximately as in 
\cite{navascues2008convergent, doherty2008quantum, HM04}). Optimality implies there exists a state $\psi$ making $$0= (\omega^* -\Phi_\game )\psi =SOS\psi.$$
Thus we get algebraic equations $s_j \psi=0$
corresponding to $SOS = \sum_j s_j^* s_j$.
For some games it is possible to show from these equations have  only one solution (that is, a unique $*$-representation)
and from this  show that the game has a unique optimal
strategy.

\sec{Linear Systems Games}
\label{sec:LinearSysGames}

Linear systems games are some of the first games to have their perfect commuting operator strategies characterized algebraically. This characterization, given in \cite{cleve2017perfect}, shows that perfect commuting operator strategies for linear systems games arise from representations of a group called the solution group of the game. As a result of this characterization, deciding existence of perfect commuting operator strategies for linear systems games was shown to be equivalent to solving an instance of the word problem. 

In this section, we reconcile this characterization of linear systems games with the \NullSS framework described in this paper. In particular, we derive the linear systems game characterization given in \cite{cleve2017perfect} using \Cref{thm:perfect_unitary_games} as a starting point. 

To begin, we recap the definition of linear systems games. 

\begin{defn}
A linear systems game $\game$ is a two player game based on a system of $m$ linear equations in $n$ variables computed mod $r$. A question to Alice is an integer $i \in [m]$ selecting an equation in the linear system equation. A question to Bob selects a variable $j \in [n]$. Alice's response consists of a vector $\vec{a}_i = (a_{i_1}, a_{i_2}, \ldots, a_{i_k})~\in~[r]^k$ containing values for all the variables contained in the $i$-th equation, with $i_1, ... i_k$ indexing the variables. Bob's response is a single variable $b_j \in [r]$. The players win each round provided variables $(a_{i_1}, a_{i_2}, \ldots, a_{i_k})$ satisfy the $i$-th equation and $a_{i_t} = b_j$ whenever 
$i_t=j$.
\end{defn}

\newcommand{\varset}{T}
Following the notation laid out in \Cref{ssec:technical_definitions} note the game algebra $\uGA$ is generated by elements $\proja{i}{\vec{a}_i}{1}$ and $\proja{j}{b_j}{2}$ with $\vec{a}_j$ ranging over all possible vectors of responses to question $i$ and $b_j$ ranging over all possible values Bob can give for variable $b_j$. 

Next, we identify elements of $\uGA$ which correspond to the value Alice gives to a single variable in the system of equations. For any $i \in [m]$ let $\varset(i) = \{t_1, t_2, \ldots, t_k\}$ list all the variables contained in the $i$-th equation in the system of equations associated with the game. Then, for any $t \in \varset(i)$ and $r' \in [r]$ define 
\begin{align}
     \proja{i}{a_{t} = r'}{1} = \sum_{\vec{a}_i : a_{t} = r'} \proja{i}{\vec{a}_i}{1}.
\end{align}
It is easy to check that the elements $\proja{i}{a_{t} = r'}{1}$ also satisfy the relations of projectors, 
\begin{subequations}
\begin{align}
    \left(\proja{i}{a_{t} = r'}{1}\right)^2 &= \left(\proja{i}{a_{t} = r'}{1}\right)^* = \proja{i}{a_{t} = r'}{1}, \\
    \sum_{r' \in [r]} \proja{i}{a_{t} = r'}{1} &= 1.
\end{align}
\end{subequations}
Additionally, orthogonality of the $\proja{i}{\vec{a}_i}{1}$ elements gives that for any $i \in [m]$, $t_1, t_2 \in \varset(i)$, and $r_1, r_2 \in [r]$ we have 
\begin{align}
    \proja{i}{a_{t_1} = r_1}{1}\proja{i}{a_{t_2} = r_2}{1}  =  \proja{i}{a_{t_2} = r_2}{1}\proja{i}{a_{t_1} = r_1}{1} =  \sum_{\substack{\vec{a}_i: a_{t_1} = r_1\\ \text{ and } a_{t_2} = r_2}}\proja{i}{\vec{a}_i}{1} \label{eq:ls_proj_commute}
\end{align}
Finally we note that, also by orthogonality of the   $\proja{i}{\vec{a}_i}{1}$ elements we have for any $i$ with $\varset(i) = \{t_1, \ldots, t_k\}$ and $\vec{a}_i = (r_{t_1}, r_{t_2}, \ldots, r_{t_k})$ that
\begin{align}
    \prod_{t \in \varset(i)} \proja{i}{a_{t} = r_{t}}{1} = \proja{i}{\vec{a}_i}{1}. \label{eq:prod_of_lsprojectors}
\end{align}

From here we can define cyclic unitary generators analogously to in \Cref{ssec:Algebraic_Picture}, with
\begin{subequations}
\begin{align}
\lcycla{t}{1}{j} &:= \sum_{r'= 1}^r \exp(\frac{2\pi r' i}{r}) \proja{j}{a_{t} = r'}{1},\\
\cycla{j}{2} &:= \sum_{a=1}^r \exp(\frac{2\pi a i}{r}) \proja{j}{a}{2}.
\end{align}
\end{subequations}

\newcommand{\lsgroup}{G_{\textrm{ls}}}

These unitary generators satisfy the same relations as the generators defined in \Cref{ssec:Algebraic_Picture}, namely,
\begin{subequations}
\begin{align}
    \left( \lcycla{t}{1}{i} \right)^r &= \left(\cycla{j}{2}\right)^r = 1, \\
    \left( \lcycla{t}{1}{i} \right)^*\lcycla{t}{1}{i} &= \left(\cycla{j}{2}\right)^*\cycla{j}{2} = 1,\\
    \lcycla{t}{1}{i}\cycla{j}{2} &= \cycla{j}{2}\lcycla{t}{1}{i} \;\;\;\; \forall \; i \in [m], \, j,t \in [n]. 
\end{align}
In addition, \Cref{eq:prod_of_lsprojectors} gives that 
\begin{align}
    \lcycla{t_1}{1}{i} \lcycla{t_2}{1}{i}  = \lcycla{t_2}{1}{i} \lcycla{t_1}{1}{i} \;\;\;\; \forall \; i \in [m], \, t_1, t_2 \in \varset(i).  \label{eq:ls_cls_com_reln}
\end{align}
\end{subequations}
We can then define $\lsgroup$ to be the group generated by the $\lcycla{t}{1}{i}$ and $\cycla{j}{2}$.

Now we show that linear systems games are \ugames.

\begin{thm} \label{thm:LSGames_Toric}
Let $\game_{ls}$ be a linear systems game based on a system of $m$ equations. For any $j \in [m]$ write the $j$-th equation of the system as 
\begin{align}
    \sum_{t \in \varset(j)} d^{(j)}_t y_t = s_j \pmod{r}
\end{align}
where $\varset(j)$ indexes all the variables $y_t$ contained in the $j$-th equation of the system and $d_t^{(j)}, s_j \in [r]$. Then  $\game_{ls}$ is a \ugame, \dby the elements 
\begin{align}
    \dset = \left\{ \exp( - \frac{2 \pi s_j i}{r}) \prod_{t \in \varset(j)} \left(\lcycla{t}{1}{j}\right)^{d_t^{(j)}} - 1\right\}_{j \in [m]} \bigcup\  \left\{ \lcycla{t}{1}{j} \left(\cycla{t}{2}\right)^* - 1\right\}_{t \in [n], j \in [m]}.
\end{align}
\end{thm}

\begin{proof}
First we note that expanding out the product and applying \Cref{eq:prod_of_lsprojectors} gives
\begin{align}
    \prod_{t \in \varset(j)} \left(\lcycla{t}{1}{j}\right)^{d_t^{(j)}} &= 
    \prod_{t \in \varset(j)} \left(\sum_{r'=1}^r \exp(\frac{2\pi d_t^{(j)} r' i}{r}) \proja{j}{a_{t} = r'}{1}  \right)\\
    &=\sum_{\vec{a}_j} \exp( \frac{2 \pi i}{r} \sum_{t \in \varset(j)} d_t^{(j)} a^{(j)}_t ) \proja{j}{\vec{a}_j}{1} \label{eq:prod_of_ls_cycla}
\end{align}
where we wrote $\vec{a_j} = (a^{(j)}_{t_1}, a^{(j)}_{t_2}, \ldots, a^{(j)}_{t_k})$ for $\{t_1, t_2, \ldots, t_k\} \in \varset(j)$. Then define 
\begin{align}
    A(j) = \left\{\vec{a}_j : \sum_{t \in \varset(j)} d_t^{(j)} a^{(j)}_t = s_j \pmod r \right\}
\end{align}
to be the collection of winning responses Alice can send to question $j$. Then as a consequence of \Cref{eq:prod_of_ls_cycla}, we have for any commuting operator strategy $(\rep,\state)$ that the condition
\begin{subequations}
\begin{align}
    \rep\left( \prod_{t \in \varset(j)} \left(\lcycla{t}{1}{j}\right)^{d_t^{(j)}} \right)\state &= \exp(\frac{2 \pi  i s_j}{r}) \state 
\end{align} 
is equivalent to 
\begin{align}
    \rep\left(\sum_{\vec{a}_j \in A(j)} \rep( \proja{j}{\vec{a}_j}{1} ) \right)\state &= \state.
\end{align}
\end{subequations}
Hence the condition 
\begin{align}
    \rep\left(\exp(- \frac{2 \pi s_j i}{r}) \prod_{t \in \varset(j)} \left(\lcycla{t}{1}{j}\right)^{d_t^{(j)}} - 1\right)\state = 0
\end{align}
for all $j \in [m]$ ensures that Alice's responses in the game $\game_{ls}$ are always winning. Similarly, the condition 
\begin{align}
   \rep\left( \lcycla{t}{1}{j} \left(\cycla{t}{2}\right)^* \right) \state = \state
\end{align}
ensures that 
\begin{align}
    \sum_{r' \in [r]}\rep \left(\proja{j}{a_{t} = r'}{1}\proja{t}{r'}{2}\right)\state = \state
\end{align}
and thus Bob's responses are also always winning. Thus it is clear that $\game_{ls}$ is \dby $\dset$.

It remains to show that all the elements of $\dset$ are of the correct form. But we have that the elements $\lcycla{t}{1}{j}$ and $\cycla{2}{t}$ generate the elements $\proja{j}{a_{t} = r'}{1}$ and $\proja{j}{a}{2}$ through the same inverse transformation as given for the cyclic unitary generators in \Cref{sssec:uGA}. And we also know that the elements $\proja{j}{a_{t} = r'}{1}$ generate the elements $\proja{j}{\vec{a}_i}{1}$ by \Cref{eq:prod_of_lsprojectors}. Then $\uGA = \mathbb{C}[\lsgroup]$ and the result is clear. 
\end{proof}

\Cref{thm:LSGames_Toric} reduces the question of whether or not a game has a perfect commuting operator strategy to an instance of the subgroup membership problem. The next theorem lets us reduce further to the standard formulation in terms of the word problem. We prepare with a few definitions, following wherever possible the conventions laid out in \Cref{sssec:tgames_smp}.

\newcommand{\lsaltgp}{\lsgroup^{\textrm{aug}}}
\newcommand{\lsclgp}{\clausegp_{\textrm{ls}}}
\newcommand{\lsnclgp}{N_{\textrm{ls}}}
\newcommand{\betaelt}{\zeta}

\begin{defn}
Let $\game_{ls}$ be a linear systems game based on a system of $m$ equations defined as in \Cref{thm:LSGames_Toric} and introduce to shorthand $\betaelt = \exp(-\frac{2 \pi i}{r})$ to simplify notation. Then define the following groups:

\begin{enumerate}[\rm(1)]
    \item $\lsaltgp$ to be the subgroup of $\mathbb{C}[G_{ls}]$ generated by the elements $G_{ls} \cup \{\betaelt\}$;
    \item $\lsclgp < \lsaltgp$ to  be the subgroup of $\lsaltgp$ generated by the set of elements 
    \begin{align}
        \left\{\betaelt^{s_i} \prod_{t \in \varset(i)} \left(\lcycla{t}{1}{i}\right)^{d_t^{(i)}}\right\}_{i \in [m]} \bigcup\ \left\{\cycla{t}{2}\left(\lcycla{t}{1}{i}\right)^{-1} \right\}_{i \in [m], t \in [r]};
    \end{align}
    \item $\lsaltgp(2)$ to be the subgroup of $\lsaltgp$ generated by the elements 
    \begin{align}
        \{\cycla{t}{2}\}_{t \in [r]} \cup \{\betaelt\};
    \end{align} 
    \item $\lsnclgp(2)$ to be the normal subgroup of $\lsaltgp(2)$ generated by the elements 
    \begin{align}
        \left\{\betaelt^{s_i} \prod_{t \in \varset(i)} \left(\cycla{t}{2}\right)^{d_t^{(i)}}\right\}_{i \in [m]} \bigcup\ \left\{\cycla{t}{2}\cycla{t'}{2}\cycla{t}{2}^{-1}\cycla{t'}{2}^{-1}\right\}_{t, t' \in \varset(i), i \in [m]}.
    \end{align}
\end{enumerate}

\end{defn}

\begin{thm}
A linear system game $\game_{ls}$, defined as in \Cref{thm:LSGames_Toric}, has a perfect commuting operator strategy iff any of the equivalent conditions are satisfied:
\begin{enumerate}[\rm(i)]
\item $\betaelt^{s} \notin \lsclgp$ for all $s \in [r-1]$;
\item $\betaelt^{s} \notin \lsnclgp(2)$ for all $s \in [r-1]$;
\item Letting $[\betaelt^{s}]$ denote the image of  $\betaelt^{s} \in \lsaltgp(2)$ in the group $\lsaltgp(2)/\lsnclgp(2)$, then $[\betaelt^{s}] \neq 1$ for all $s \in [r-1]$. 
\end{enumerate}
\end{thm}

\begin{proof}
We first note that condition (i) is equivalent to the existence of a perfect commuting operator strategy by \Cref{thm:LSGames_Toric} and \Cref{thm:perfect_unitary_games}.

To show  (i) $\Rightarrow$ (ii) we show $\lsnclgp(2)$ is contained in $\lsclgp$. First note that for any $i \in [m]$ we have 
\begin{align}
    &\left(\betaelt^{s_i} \prod_{t \in \varset(i)} \left(\lcycla{t}{1}{i}\right)^{d_t^{(i)}}\right) \prod_{t \in \varset(i)} \left( \left(\lcycla{t}{1}{i}\right)^{-1} \cycla{t}{2}\right)^{d_t^{(i)}} = \betaelt^{s_i} \prod_{t \in \varset(i)} \left(\cycla{t}{2}\right)^{d_t^{(i)}} \in \lsclgp
\end{align}
and, for any $i \in [m]$ and $t_1, t_2 \in \varset(i)$ we can multiply together generators of the group $\lsclgp$ to find
\begin{subequations}
\begin{align}
   &\left(\cycla{t_1}{2}\left(\lcycla{t_1}{1}{i}\right)^{-1} \right)
   \left(\cycla{t_2}{2}\left(\lcycla{t_2}{1}{i}\right)^{-1} \right)
   \left(\cycla{t_1}{2}\left(\lcycla{t_1}{1}{i}\right)^{-1} \right)^{-1}
   \left(\cycla{t_2}{2}\left(\lcycla{t}{1}{i}\right)^{-1} \right)^{-1} \nonumber \\
   &\hspace{50pt}= \left(\cycla{t_1}{2}\cycla{t_2}{2}\cycla{t_1}{2}^{-1}\cycla{t_2}{2}^{-1}\right) \left(\left(\lcycla{t_1}{1}{i}\right)^{-1}\left(\lcycla{t_2}{1}{i}\right)^{-1} \lcycla{t_1}{1}{i}
   \lcycla{t_2}{1}{i}\right).
   \end{align}
   Now we cancel 
    the $\lcycla{t_1}{1}{i}$ terms using \Cref{eq:ls_cls_com_reln} to get
   \begin{align} 
   \hspace{50pt} \cycla{t_1}{2}\cycla{t_2}{2}\cycla{t_1}{2}^{-1}\cycla{t_2}{2}^{-1} \in \lsclgp.
\end{align}
\end{subequations}
Finally, we note that for any elements $w(2) \in \lsclgp \cap \lsaltgp(2)$ and $\cycla{t}{2}$ generating $\lsaltgp(2)$ we have 
\begin{align}
    \left(\cycla{t}{2}\left(\lcycla{t}{1}{i}\right)^{-1}\right) w(2) \left(\cycla{t}{2}\left(\lcycla{t}{1}{i}\right)^{-1} \right)^{-1} = \cycla{t}{2} w(2) \cycla{t}{2}^{-1} \in \lsclgp
\end{align}
and thus, the normal closure in $\lsaltgp(2)$ of any element in $\lsclgp \cap \lsaltgp(2)$ is also contained in $\lsclgp$ (recall that the element $\betaelt$ is central). We conclude that $\lsnclgp(2) \subseteq \lsclgp$, as desired. 

To show that (ii) $\Rightarrow$ (i) we assume $\betaelt^s \in \lsclgp$ for some $s \in [r-1]$. 
Then, there exists elements $w(1)_1, \ldots, w(1)_L$ and $v(1,2)_1, \ldots, v(1,2)_{L+1}$ with each $w(1)_\ell $ equal to a product of elements or inverses of elements in the set
\begin{align}
    \left\{\betaelt^{s_i} \prod_{t \in \varset(i)} \left(\lcycla{t}{1}{i}\right)^{d_t^{(i)}}\right\}_{i \in [m]}\label{eq:w_set}
\end{align}
each $v(1,2)_{\ell'} $ equal to a product of elements or inverses of elements in the set 
\begin{align}
\left\{\cycla{t}{2}\left(\lcycla{t}{1}{i}\right)^{-1} \right\}_{i \in [m], t \in [r]} \label{eq:v_set}
\end{align}
and 
\begin{align}
    \betaelt^s = v(1,2)_1 w(1)_1 v(1,2)_2 w(1)_2 \cdots v(1,2)_L w(1)_L v(1,2)_{L+1}.
\end{align}
Then let the elements $w(2)_\ell$ for $\ell \in \{1,2,...,L\}$ be obtained 
by choosing some product of elements drawn from the set defined in \Cref{eq:w_set} and equal to $w(1)_\ell$, then making the replacement
\begin{align}
\betaelt^{s_i} \prod_{t \in \varset(i)} \left(\lcycla{t}{1}{i}\right)^{d_t^{(i)}} \rightarrow \betaelt^{s_i} \prod_{t \in \varset(i)} \left(\cycla{t}{2}\right)^{d_t^{(i)}}
\end{align}
to each element in the product.\footnote{Note we have not shown this definition of $w(2)_\ell$ is unique, i.e. we have not shown that this replacement defines a homomorphism, but we will not need to for the current proof.} Similarly let elements $v(1)_1, \ldots, v(1)_{L+1}$ be obtained from $v(1,2)_\ell$ by the replacement 
\begin{align}
    \cycla{t}{2}\left(\lcycla{t}{1}{i}\right)^{-1} \rightarrow \left(\lcycla{t}{1}{i}\right)^{-1}
\end{align}
and elements $v(2)_1, \ldots, v(2)_{L+1}$ be obtained
via the replacement \begin{align}
    \cycla{t}{2}\left(\lcycla{t}{1}{i}\right)^{-1} \rightarrow \cycla{t}{2}
\end{align}
applied to some fixed products of elements from \Cref{eq:v_set} equal to $v(1,2)_\ell$.

Then commuting $\cycla{t}{1}$ and $\cycla{t}{2}$
 elements gives
\begin{align}
  \betaelt^{s}  & = v(1,2)_1 w(1)_1 v(1,2)_2 w(1)_2 \cdots v(1,2)_L w(1)_L v(1,2)_{L+1} \nonumber\\
    &= v(1)_1 w(1)_1 v(1)_2 w(1)_2 \cdots v(1)_L w(1)_L v(1)_{L+1} v(2)_1 v(2)_2 \cdots v(2)_{L+1}. \\
\end{align}
From this we conclude 
\begin{align}
    v(2)_1 v(2)_2 \cdots v(2)_{L+1} = 1
\end{align}
since the word $v(2)_1 v(2)_2 \cdots v(2)_{L+1}$ contains only $\cycla{t}{2}$ elements (and no product of $\cycla{t}{2}$ elements and their inverses can be to equal $\betaelt$ or to any product of $\cycla{t}{1}$ elements) and hence
\begin{align}
    v(1)_1 w(1)_1 v(1)_2 w(1)_2 \cdots v(1)_L w(1)_L v(1)_{L+1} = \betaelt^s.
\end{align}
But then we also have 
\begin{align}
    v(2)_1 w(2)_1 v(2)_2 w(2)_2 \cdots v(2)_L w(2)_L v(2)_{L+1} = \betaelt^s
\end{align}
since the relations between $\cycla{t}{2}$ elements are the same as those between the $\lcycla{t}{1}{i}$ elements. Then the calculation
\begin{subequations}
\begin{align}
    \betaelt^s &= \left(\prod_{\ell \in [L]} v(2)_{\ell} w(2)_{\ell}\right) v(2)_{L + 1} \\
    &= \left(\prod_{\ell \in [L]} \left( \prod_{\ell' < \ell} v(2)_{\ell'} \right) w(2)_{\ell} \left( \prod_{\ell' < \ell} v(2)_{\ell'} \right)^{-1} \right) v(2)_1 v(2)_2 \cdots v(2)_{L + 1} \\
    &= \left(\prod_{\ell \in [L]} \left( \prod_{\ell' < \ell} v(2)_{\ell'} \right) w(2)_{\ell} \left( \prod_{\ell' < \ell} v(2)_{\ell'} \right)^{-1} \right) \in \lsnclgp(2)
\end{align}
\end{subequations}
shows $\betaelt^s \in \lsnclgp(2)$, as desired. 

Finally we note (ii) $\Leftrightarrow$ (iii) by the 
definition of a quotient group.
\end{proof}

\sec{\NullSS applied to Synchronous Games}
\label{sec:Synchronous}

As discussed in the introduction, a two player game is called synchronous if it includes ``consistency-checks'' where Alice and Bob are sent the same question and win iff they send the same response. Below we give a formal definition of synchronous games using the scoring function and probability distribution description of games introduced in \Cref{sssec:standard_games_notation}. 

\begin{defn} \label{defn:synch_games}
A two player, $n$ question, $m$ response game $\game$ defined by scoring function $V$ and question distribution $\mu$ is called a \df{synchronous game} iff, for all $i \in [n]$ we have
\begin{align}
    \mu(i,i) > 0 
\end{align}
and 
\begin{align}
    V(a,b | i, i) = \delta_{a,b} 
\end{align}
where $\delta_{a,b}$ denotes the Kronecker delta function. 
\end{defn}

An immediate consequence of \Cref{defn:synch_games} in terms of the 
notion of
invalid \detset introduced in \Cref{sssec:valid_invaid_response_det_sets} is the following claim.\footnote{We note this claim can hold true for a larger class of games than synchronous ones. In most of this section the only property of synchronous games which will be used is that they satisfy \Cref{eq:synch_detset_membership}. Thus, most of the techniques of this section apply to a slightly larger class of games than synchronous ones.} 

\begin{claim} \label{claim:sych_invalidpoly}
Let $\game$ be a synchronous game and $\invalidpoly$ be the invalid elements of $\game$. Then we have 
\begin{align}
    \proja{i}{a}{1}\proja{i}{b}{2} \in \cLI(\invalidpoly) \label{eq:synch_detset_membership}
\end{align}
for all $i \in [n]$ and $a, b \in [m]$ with $a \neq b$. 
\end{claim}

\begin{proof}
Immediate since for all for all $i \in [n]$ and $a, b \in [m]$ with $a \neq b$ we have $ V(a,b | i, i) = 0$ and $\mu(i,i) > 0$ by definition of a synchronous game, hence $\proja{i}{a}{1}\proja{i}{b}{2} \in \invalidpoly$ by the definition of $\invalidpoly$.
\end{proof}

Synchronous games can be studied using the standard techniques of this paper, where we consider representations of the algebra $\uGA$ which satisfy directional zero constraints.
However, Paulsen with various
collaborators 
\cite{paulsen2016estimating} 
\cite[Theorem 3.2]{helton2017algebras},
found an equivalent simpler formulation using a `smaller'
algebra which we denote $\uGA(1)$,
and using hard zeroes arising from an ideal $\ide(\sbr(1))$ of $\uGA(1)$.
Also the papers \cite{paulsen2016estimating,kim2018synchronous}
  show that
 the synchronous value of a game is given by
 the trace of a bilinear function
 on $\uGA(1)$.
The first goal in this section  is to show
how this
reconciles with our \NullSSs.
 This is the substance of 
\Cref{thm:syncRep} in \Cref{sec:syncAlg}.

Next in   \Cref{sec:traceNSS} we turn to the tracial \NullSS and recall that the theory of
Null- and Positivstellens\"atze appropriate to tracial situations go back to \cite{klepConnes}
and is well developed in
subsequent papers which are
summarized in \cite[Chapter 5]{klepBook}; see also \cite{klepPTrace}.

Finally, in \Cref{sec:GBExamples1} we demonstrate (on a graph coloring problem) a computer algorithm based on Gr\"obner Bases plus \NullSS.

\ssec{Synchronous two player games in terms of an algebra}
\label{sec:syncAlg}

In what follows let:
\begin{enumerate}[\rm(1)]
    \item $\uGA$ be the universal game algebra;
    \item $\cLI(\br)$ be the left ideal of $\uGA$ generated by the set $\left\{\prod_\alpha e(\alpha)^{i(\alpha)}_{a(\alpha)} : (\vec{i},\vec{a}) \in \br\right\}$;
    \item $\cLI(\sbr)$ be the left ideal of $\uGA$ generated by the set $$\left\{\prod_\alpha \proja{i(\alpha)}{a(\alpha)}{1} : (\vec{i},\vec{a}) \in \br\right\} \cup \left\{\prod_\alpha \proja{i(\alpha)}{a(\alpha)}{2} : (\vec{i},\vec{a}) \in \br\right\} \cup \left\{\proja{i}{a}{1} - \proja{i}{a}{2} : i \in [n], a \in [m] \right\};$$
    \item $\uGA(1)$ be the subalgebra of $\uGA$ generated by $\proja{i}{a}{1}$ only; 
    \item $\ide(\sbr(1))$ be the two sided ideal of $\uGA(1)$ generated by $\left\{\prod_\alpha \proja{i(\alpha)}{a(\alpha)}{1} : (\vec{i},\vec{a}) \in \br\right\}.$
\end{enumerate}

Now we state the 
objective of this section.

\begin{thm}
\label{thm:syncRep}
A synchronous game characterized by a set $\br$ of invalid responses has a perfect commuting operator strategy iff any of the equivalent conditions are satisfied:
\begin{enumerate}[\rm(i)]
    \item\label{it:cond1} There exists a $*$-representation $\pi: \uGA \rightarrow \cB(\cH)$ and a state $\psi \in \cH$ satisfying 
    \begin{align}
        \pi(\cLI(\br)) \psi = \{0\};
    \end{align}
    \item[\rm(i)']\label{it:cond1'}
    $\zd^{\rm re,\uGA}(\br)\neq\varnothing$;
    \item\label{it:cond2} There exists a $*$-representation $\pi: \uGA \rightarrow \cB(\cH)$ and a state $\psi \in \cH$ satisfying 
    \begin{align}
        \pi(\cLI(\sbr)) \psi = \{0\};
    \end{align}
    
    \item[\rm(ii)']\label{it:cond2'}
    $\zd^{\rm re,\uGA}(\sbr)\neq\varnothing$;
    
    \item\label{it:cond3} There exists a $*$-representation $\pi': \uGA(1) \rightarrow \cB(\cH)$ and a tracial state $\psi \in \cH$ satisfying 
    \begin{align}
        \pi'(\ide(\sbr(1)) \psi = \{0\};
    \end{align}

\def\cW{\mathcal W}
    \item\label{it:cond4} There exists a $*$-representation $\pi'$ of $\uGA(1)$ mapping into a tracial von Neumann algebra
    $\cW\subseteq \cB(\cH)$ satisfying
    \begin{align}
        \pi'(\ide(\sbr(1))) = \{0\};
    \end{align}
 \end{enumerate}
\end{thm}

The proof of (i), (ii), (iii) is based 
on several lemmas which we now give.
The core of the proofs of the lemmas come from 
\cite{paulsen2016estimating}
but we include a self contained account for clarities sake.

\begin{lem} \label{lem:LI=synchLI}
If the set of invalid responses $\br$ comes from a synchronous game, we have $\cLI(\br) = \cLI(\sbr)$. 
\end{lem}
\begin{proof}
We first show that $\proja{i}{a}{1} - \proja{i}{a}{2} \in \cLI(\br)$. Because $\br$ corresponds to a synchronous game, we have 
\begin{align}
    \proja{i}{a}{1}\proja{i}{b}{2} \in \br \subseteq \cLI(\br)
\end{align}
for all $b \neq a$. Then we also have 
\begin{subequations}
\begin{align}
   \cLI(\br) \ni   \sum_{b\neq a} \proja{i}{a}{1}\proja{i}{b}{2} &=  \proja{i}{a}{1} \sum_{b\neq a} \proja{i}{b}{2} 
    = \proja{i}{a}{1} (1 - \proja{i}{a}{2}) \\
    &= \proja{i}{a}{1} - \proja{i}{a}{1}\proja{i}{a}{2} 
    .\label{eq:subtract1}
\end{align}
\end{subequations}
Similarly, summing over the $e(1)$ terms (and recalling that $e(1)$ and $e(2)$ commute) gives: 
\begin{align}\label{eq:subtract2}
     \sum_{b\neq a} \proja ib1\proja{i}{a}{2} 
    &= \proja{i}{a}{2} (1 - \proja{i}{a}{1}) 
    = \proja{i}{a}{2} - \proja{i}{a}{1}\proja{i}{a}{2} \in \cLI(\br).
\end{align}
Subtracting \Cref{eq:subtract2} from \Cref{eq:subtract1} gives 
\begin{align}
    \proja{i}{a}{1} - \proja{i}{a}{2} \in \cLI(\br).
\end{align}

Next note that for any $\proja{i}{a}{1} \proja{j}{b}{2} \in \br$ we have 
\begin{align}
    \proja{i}{a}{1} \proja{j}{b}{1} = \proja{i}{a}{1} \big(\proja{j}{b}{1} - \proja{j}{b}{2}\big) + \proja{i}{a}{1} e(2)^j_b\in \cLI(\br).
\end{align}
A similar argument shows  $\proja{i}{a}{2} \proja{j}{b}{2} \in \cLI(\br)$. Then all the generators of $\cLI(\sbr)$ are contained in $\cLI(\br)$ and we conclude $\cLI(\sbr) \subseteq \cLI(\br)$. 

To prove the converse we observe that for any $\proja{i}{a}{1} \proja{j}{b}{2} \in \br$ we have 
\begin{align}
    \proja{i}{a}{1} \proja{j}{b}{2} = \proja{i}{a}{1} \proja jb1 - \proja{i}{a}{1} \big(\proja{j}{b}{1} - \proja{j}{b}{2}\big) \in \cLI(\sbr),
\end{align}
which gives $\cLI(\br) \subseteq \cLI(\sbr)$ and completes the proof. 
\end{proof}

\begin{lem}\label{lem:LI(1)inSynchLI}
$\ide(\sbr(1)) \subseteq \cLI(\sbr)$.
\end{lem}
\begin{proof}
First consider a monomial $m(1) = e(1)^{i_1}_{a_1}  \cdots e(1)^{i_t}_{a_t} \in \uGA(1)$ and note 
\begin{subequations}
\begin{align}
    e(1)^{i_1}_{a_1}  \cdots e(1)^{i_t}_{a_t}
    &=  e(1)^{i_1}_{a_1} \cdots e(1)^{i_{t-1}}_{a_{t-1}} (e(1)^{i_t}_{a_t} - e(2)^{i_t}_{a_t}) +  e(1)^{i_1}_{a_1}  \cdots e(1)^{i_{t-1}}_{a_{t-1}} e(2)^{i_t}_{a_t} \\
    &=  e(1)^{i_1}_{a_1}  \cdots e(1)^{i_{t-1}}_{a_{t-1}} (e(1)^{i_t}_{a_t} - e(2)^{i_t}_{a_t}) +  e(2)^{i_t}_{a_t} e(1)^{i_1}_{a_1} \cdots e(1)^{i_{t-1}}_{a_{t-1}},
\end{align}
\end{subequations}
whence 
\[
    e(1)^{i_1}_{a_1} \cdots e(1)^{i_t}_{a_t} - e(2)^{i_t}_{a_t} e(1)^{i_1}_{a_1} \cdots e(1)^{i_{t-1}}_{a_{t-1}} \in \cLI(\sbr).
\]
Applying this  inductively we see 
\begin{subequations}
\begin{align}
    &e(1)^{i_1}_{a_1} \cdots e(1)^{i_t}_{a_t} - e(2)^{i_t}_{a_t} e(1)^{i_1}_{a_1} \cdots e(1)^{i_{t-1}}_{a_{t-1}} \\
    &\hspace{10pt} + e(2)^{i_t}_{a_t} \left(e(1)^{i_1}_{a_1} \cdots e(1)^{i_{t-1}}_{a_{t-1}}  - e(2)^{i_{t-1}}_{a_{t-1}}e(1)^{i_1}_{a_1} \cdots e(1)^{i_{t-2}}_{a_{t-2}}\right) \\
    &\hspace{20pt} + e(2)^{i_t}_{a_t}e(2)^{i_{t-1}}_{a_{t-1}} \left(e(1)^{i_1}_{a_1} \cdots e(1)^{i_{t-1}}_{a_{t-1}}  - e(1)^{i_{t-2}}_{a_{t-2}} e(1)^{i_1}_{a_1} \cdots e(1)^{i_{t-3}}_{a_{t-3}}\right) \\
    &\hspace{30pt} + \cdots \\
    &= e(1)^{i_1}_{a_1} \cdots e(1)^{i_t}_{a_t} - e(2)^{i_t}_{a_t}e(2)^{i_{t-1}}_{a_{t-1}} \cdots e(2)^{i_1}_{t_1} \in \cLI(\sbr).
\end{align}
\end{subequations}
We can write this last expression compactly as $m(1) - m(2)^* \in \cLI(\sbr)$, where we understand $m(2)$ to be the polynomial $m(1)$ with all $e(1)$ replaced by $e(2)$. By linearity this observation immediately extends to general polynomials, so we have $p(1) - p(2)^* \in \cLI(\sbr)$ for any polynomial $p(1)$ formed entirely from $e(1)$. 

For any polynomial $p \in \ide(\sbr(1))$ we can write $p = s(1)t(1)$ where $s(1) \in \cLI(\sbr)$ and $t(1) \in \uGA(1)$ is an arbitrary polynomial consisting only of $e(1)$ generators. Then we observe:
\begin{align}
    p = s(1)t(1) = s(1)(t(1) - t(2)^*) + t(2)^* s(1) \in \cLI(\sbr)
\end{align} 
and we are done. 
\end{proof}

Also important in this section are tracial linear mappings on an algebra, defined to be linear mappings $\tau: \algebra \rightarrow \mathbb{C}$ satisfying 
\begin{align}
    \tau(ab) = \tau(ba) 
\end{align}
for all $a,b \in \algebra$. A state $\state$ is called a tracial state (for some operator algebra $\cA$) if the linear mapping it induces is tracial, so
\begin{align}
    \state^* a b \state = \state^* b a \state
\end{align}
for all operators $a,b \in \algebra$. We show, using an argument very similar to the one given in the proof of \Cref{lem:LI(1)inSynchLI}, that any linear mapping vanishing on $\cLI(\sbr)$ and $\cRI(\sbr)$ must be tracial on $\uGA(1)$. 

\begin{lem} \label{lem:vanishing_mappings_are_tracial}
Given a linear mapping $\tau : \uGA \rightarrow \mathbb{C}$ on $\uGA$ satisfying 
\begin{align}
    \tau(\cLI(\sbr)) =0= \tau(\cRI(\sbr)) ,
\end{align}
then $\tau$ is tracial on $\uGA(1)$;
i.e. for any $a,b \in \uGA(1)$,  
\begin{align}
    \tau(ab) = \tau(ba).
\end{align}
\end{lem}

Note: If $\tau$ is symmetric, then 
since $\sbr$ is a $*$-closed set, we have
 \[\tau(\cLI(\sbr)) =0 \implies \tau(\cRI(\sbr)) =0.\]

\begin{proof}
Consider monomials $w,\tilde{w} \in \uGA(1)$. Since $w$ is generated by $\proja{i}{a}{1}$, we can  write $w = w' \proja{i}{a}{1}$ and then observe 
\begin{subequations}
\begin{align}
    \tilde{w} w' \left(\proja{i}{a}{1} - \proja{i}{a}{2}\right) & \in \cLI(\sbr), \label{eq:synch_proj_LI_vanishes} \\
    \left(\proja{i}{a}{1} - \proja{i}{a}{2}\right) \tilde{w} w' &\in \cRI(\sbr). \label{eq:synch_proj_RI_vanishes} 
\end{align}
\end{subequations}
Then
\begin{align}
    \tau( \tilde{w} w' \proja{i}{a}{1} ) &= \tau( \tilde{w} w' \proja{i}{a}{2} )  
    =\tau(\proja{i}{a}{2} \tilde{w} w' ) 
    =\tau(\proja{i}{a}{1}  \tilde{w} w' ) 
\end{align}
where we use \Cref{eq:synch_proj_LI_vanishes} for the first equality, that elements of $\uGA(1)$ and $\uGA(2)$ commute for the second equality, and  \Cref{eq:synch_proj_RI_vanishes} for the equality. Repeating this argument shows that for any elements $w, \tilde{w} \in \sbr(1)$ we have 
\begin{align}
    \tau(w \tilde{w}) = \tau(\tilde{w} w)
\end{align}
and the proof is complete. 
\end{proof}

Now comes the proof of  the main theorem of this section.


\def\cW{\mathcal W}
\begin{proof}[Proof of \Cref{thm:syncRep}]
\Cref{it:cond1} is equivalent to existence of a perfect commuting operator strategy by definition (see \Cref{ssec:Characterizing_perfect_games}).

\Cref{it:cond2} is equivalent to \Cref{it:cond1} by \Cref{lem:LI=synchLI}. 

To prove \Cref{it:cond2} $\Rightarrow$ \Cref{it:cond3} we let $\pi$ be any representation satisfying \Cref{it:cond2}, and define $\pi'$ to be the restriction of $\pi$ to $\uGA(1)$. Clearly $\pi' : \uGA(1) \rightarrow \cB(\cH)$ and 
\begin{align}
\pi'(\ide(\sbr(1)) \psi = \pi(\ide(\sbr(1)) \psi \subseteq \pi(\cLI(\sbr)) \psi = \{0\},
\end{align}
where the inclusion follows from \Cref{lem:LI(1)inSynchLI}. To show the state $\state$ is tracial on $\pi'(\uGA(1))$ note the linear mapping defined by
\begin{align}
    \tau(w) = \psi^* \pi(w) \psi
\end{align}
vanishes on $\cLI(\sbr)$ (and hence $\cRI(\sbr)$) by \Cref{it:cond2} and so $\tau$ is tracial on $\uGA(1)$ by \Cref{lem:vanishing_mappings_are_tracial}. 

We prove \Cref{it:cond3} $\Rightarrow$ \Cref{it:cond1}. Using $\pi',\psi$ we define the positive linear functional
\[
\ell':\uGA(1)\to\C, \quad f \mapsto \psi^* \pi'(f)\psi.
\]
Next extend $\ell'$ to a linear functional $\ell$ on $\uGA$ by mapping
a monomial
\[
w(e(1)) u(e(2)) \mapsto \ell\big(w(e(1)) u(e(1))^*\big).
\]

It is obvious that $\ell$ is symmetric in the sense that
$\ell(f^*)=\ell(f)^*$ for all $f\in\uGA(1)$. 
To check that $\ell$ is positive, let 
$f=\sum_{i,j} \beta_{ij} w_i(e(1)) u_j(e(2))\in\uGA$.
Then
\[
f^*f= \sum_{i,j}\sum_{k,l} \beta_{ij}^*\beta_{kl} w_i(e(1))^* w_k(e(1)) u_j(e(2))^* u_l(e(2)),
\]
whence
\beq\label{eq:sos1}
\ell(f^*f) = \sum_{i,j}\sum_{k,l} \beta_{ij}^*\beta_{kl} \ell' ( w_i(e(1))^* w_k(e(1)) u_l(e(1))^* u_j(e(1))).
\eeq
Set
\[
\check f=\sum_{i,j} \beta_{ij}  w_i(e(1))u_j(e(1))^* \in\uGA(1).
\]
Then
\[
\check f^*\check f= 
\sum_{i,j}\sum_{k,l} \beta_{ij}^*\beta_{kl} u_j(e(1)) w_i(e(1))^* w_k(e(1)) u_l(e(1))^*,
\]
and
\beq\label{eq:sos2}
\ell'(\check f^*\check f)= \sum_{i,j}\sum_{k,l} \beta_{ij}^*\beta_{kl} \ell'( u_j(e(1)) w_i(e(1))^* w_k(e(1)) u_l(e(1))^* ).
\eeq
Since $\ell'$ is tracial,
\[
\ell'( u_j(e(1)) w_i(e(1))^* w_k(e(1)) u_l(e(1))^* )
= 
\ell' ( w_i(e(1))^* w_k(e(1)) u_l(e(1))^* u_j(e(1))).
\]
This implies the values in \Cref{eq:sos1} and \Cref{eq:sos2} are the same, so
\[
\ell(f^*f) = \ell'(\check f^*\check f)\geq0.
\]

It remains to show that $\ell( \cLI(\br))=\{0\}$. Elements in 
$\cLI(\br)$ are linear combinations of monomials
 of the form
\beq\label{eq:eltInN}
w(e(1)) u(e(2)) e(1)^i_a e(2)^j_b = 
w(e(1)) e(1)^i_a u(e(2))  e(2)^j_b 
\eeq
with $((i,j),(a,b))\in\br$. Applying  $\ell$ to \Cref{eq:eltInN}
gives
\[
\ell(w(e(1)) e(1)^i_a u(e(2))  e(2)^j_b )
=\ell'(w(e(1)) e(1)^i_a e(1)^j_b u(e(1))^* )
\]
But $e(1)^i_a e(1)^j_b\in\br$, whence 
$w(e(1)) e(1)^i_a e(1)^j_b u(e(1))^* \in \ide(\sbr(1)) $,
so
\[
\ell'(w(e(1)) e(1)^i_a e(1)^j_b u(e(1))^* )=0,
\]
as desired. We have proved that
$\ell(-1)=-1$ and $\ell(\SOS_{\uGA}+\cLI(\br)+\cLI(\br)^*)\subseteq\R_{\geq0}$,
whence $-1\notin \SOS_{\uGA}+\cLI(\br)+\cLI(\br)^*$. Then 
\Cref{thm:AB} implies \Cref{it:cond1}.

\def\cK{\mathcal K}

We next show \Cref{it:cond3} $\Rightarrow$ \Cref{it:cond4}.
Define the Hilbert space $\check\cH:=[\pi'(\uGA(1))\psi]\subseteq\cH$.
By definition, $\pi'(\uGA(1))\check\cH\subseteq\check\cH$, so $\pi'$ induces
a $*$-representation $\check\pi':\uGA(1)\to\cB(\check\cH)$.
By construction, $\check\pi'(\ide(\sbr(1))=\{0\}$, as desired in 
\Cref{it:cond4}.

Finally, to go from \Cref{it:cond4} to \Cref{it:cond3}, 
start with the tracial von Neumann algebra
$\cW$ with trace $\tau$ as in \Cref{it:cond4} and perform a GNS construction.
There is a Hilbert space $\cK$, unit vector $\xi\in\cK$, and
a $*$-representation $\pi'':\cW\to\cB(\cK)$ so that
\[
\tau(a)= \langle \pi''(a) \xi , \xi\rangle, \quad a\in\cW.
\]
Since $\tau$ is a trace, $\xi$ is a tracial state for $\pi''(\cW)$.
Then the $*$-representation $\pi''\circ\pi':\uGA(1)\to\cB(\cK)$ together with 
$\xi\in\cK$ satisfy \Cref{it:cond3}.


\end{proof}

\ssec{Tracial Nullstellensatz}
\label{sec:traceNSS}

Now we discuss a \NullSS suited to synchronous games by virtue of handling traces.
The earliest tracial  Positivstellensatz is in 
\cite{klepConnes},
 done  in the context of the Connes' embedding conjecture.
An exposition of this and 
the (dual) tracial moment problem
is the book \cite[Chapter 5]{klepBook}
with more extensive results 
along the lines of the theorem below,
for example in 
 \cite[Proposition 4.2 and later]{klepPTrace}
 and more moment theory in
\cite{klepJOT}.
Moments with with informative quantum 
games situations are also in the recent paper 
 \cite{Russel21}.

Suppose $\cA$ is a $*$-algebra.
The \df{commutator} of $a,b\in\cA$ is
\[
[a,b]:=ab-ba,
\]
and we call $a,b$ \df{cyclically equivalent}, $a\csym b$, if
$a-b$ is a sum of commutators. Define
\[
\cycsos:=\{a\in\cA \mid \exists b\in\SOS:\, a\csym b\}.
\]
Elements of $\cycsos$ are trace positive under all
$*$-representations of $\cA$ in \trvna.

     \begin{thm}
     \label{thm:ABC}
     Suppose 
     \ben [\rm(1)]
     \item
     $\cA$ is a $*$-algebra, where $\cycsos$ is \df{archimedean} in the sense that for every $a\in\cA$ there is $\eta\in\mathbb N$ with $\eta-a^*a\in\cycsos$;
     \item $\cLI\subseteq\cA$ is a left ideal.
     \een 
      Then 
the following are equivalent:
\ben[\rm(i)]
\item
      there exist a 
      $*$-representation $\pi:\cA\to\cB(\cH)$
      and  tracial state $0\neq\psi \in \cH$ satisfying
      \beq 
      \label{eq:fX01}
      \pi(f) \psi = 0  
      \eeq 
      for all $f \in \cLI$;
\item 
there exist a 
      $*$-representation $\pi:\cA\to\cF$ into a \trvna
      $(\cF,\tau)$
      satisfying
      \beq 
      \label{eq:fX02}
      \tau(\pi(f)) = 0  
      \eeq 
      for all $f \in \cLI$;
      \item $-1 \not \in \cycsos + \cLI + \cLI^*.$
      \een
     \end{thm}
     
     \begin{proof} 
      The equivalence of Items (i) and (ii) 
is established via the GNS representation as in the proof of \Cref{thm:syncRep}.
The implication (ii) $\Rightarrow $ (iii) is easy. Namely, if
  $-1  \in \cycsos + \cLI + \cLI^*$, and $(\cF,\tau)$ as in (ii) exist,
  then 
\[-1=\tau(-1) \in \tau(\cycsos + \cLI + \cLI^*)=
\tau(\cycsos) + \tau(\cLI) + \tau(\cLI)^*= \tau(\cycsos) \subseteq\R_{\geq0},
\]
a contradiction.

For the harder side (iii) $\Rightarrow $ (ii) we only give a sketch since it is very similar
to that in the proof of \Cref{thm:AB}.   
     Suppose  $-1  \not\in \cycsos + \cLI + \cLI^*$.
By the Hahn-Banach theorem (version due to
     Eidelheit-Kakutani)
     there is a
     linear functional $L:\cA\to\C$ satisfying
      \beq\label{eq:2GNS'}  L(1)= 1, \qquad L(\cycsos + \cLI + \cLI^*) \subseteq \R_{\geq0}.
      \eeq 
     Since $\cLI$ is a subspace, the second property of \Cref{eq:2GNS'} implies $L(\cLI)=\{0\}$. Likewise,  
$\cycsos$ contains all commutators and $L(\cycsos)\subseteq\R_{\geq0}$,
whence $L(f)\geq 0$ for any $f\in\SOS_{\cA}$ and
     $L$ is tracial. Further,
$L(f)^* = L(f^*)$ for all $f\in\cA$.
     Then the GNS construction yields the desired conclusion.
     \end{proof}

\ssec{Quantum graph coloring}
\label{sec:GBExamples1}

As an example of how one can use  Theorem \ref{thm:syncRep} we look at quantum coloring of graphs. It is often formulated as a synchronous game. We start by recalling the coloring setup
then we illustrate a general approach through an example.

Given a graph $G=(V,E)$, where $V$ is the set of vertices and\
$E\subseteq V\times V$ its edges, we say $G$ admits a \df{quantum $c$-coloring} \cite{paulsen2016estimating}, if to each vertex $i\in V$ one can assign projections
$e^i_a$, $a\in\{1,\ldots,c\}$ so that 
\beq\label{eq:graph20}
\sum_a e^i_a=1,
\eeq
and for each edge $(i,j)\in E$, we have
\beq\label{eq:graph2}
e^i_a e^j_a=0,\quad \forall\, a.
\eeq

In algebraic terms, $G$ admits a quantum $c$-coloring
if the the universal game algebra
$\uGA(1)$ generated by $e(1)^i_a:=e^i_a$, where $i$
indexes the vertices and $a\in\{1,\ldots,c\}$
admit a $*$-representation $\pi$ that vanishes
on the polynomials   $e^i_a e^j_a$ from \Cref{eq:graph2}.

\sssec{Quantum graph coloring as a synchronous game}

We shall be brief here  
and refer to \cite[Section 2]{paulsen2016estimating}
for a thorough discussion of 
graph coloring and the connection to
synchronous  games.

A game is specified by a graph $G$.
The  verifier's questions amount to giving  vertex $i$  to Alice and vertex $j$ to Bob; a perfect strategy consists of a vector $\psi$ and
$k$ projectors
$\proja{i}{\vec{a}}{1}$ for Alice and
$\proja{i}{\vec{b}}{2}$ for Bob
meeting the (universal game) constraints of
  \Cref{eq:defUA}
and  
$\proja{i}{a}{1}\proja{j}{a}{2}= 0$ if vertices  $i$ and $j$ are joined by an edge.
These give the probability of Alice and Bob
answering question $(i,j)$ with colors $a,b \in 1, \dots,  c$ via the formula
$$p(a,b | i,j):= \psi^* \proja{i}{a}{1}\proja{j}{b}{2} \psi .
$$
A synchronous strategy means for $i$ any vertex, $p(a,b | i,i)=0$ if $a \not =b$.
This is equivalent to 
$\proja{i}{a}{1} \psi =\proja{i}{a}{2} \psi$ for $a=1, \dots, c$.
%
%
\Cref{thm:syncRep}
now converts solving this system of equations
 to the 
 finding tracial representations of
 $\uGA(1)$ vanishing on
 $\sbr(1)$. In our forthcoming  example, we show there are no representations much less tracial ones, so no 4-coloring exists.

\sssec{Example}

Consider the graph $G=(V,E)$ given in \Cref{fig:5}.
It is obtained from a five cycle by adding two apexes (denoted $6$ 
and $7$ in the figure), i.e.,
two vertices connected to all the other vertices.

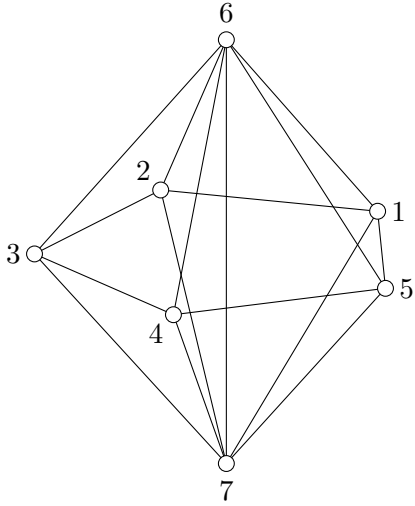
\begin{figure}[H]
\centering
\tdplotsetmaincoords{70}{-20}
\begin{tikzpicture}[scale=3,tdplot_main_coords]

\coordinate (A1) at (0.809017 , 0.262866 ,0) ;
\coordinate (A2) at (0, 0.850651 ,0) ;
 \coordinate (A3) at (-0.809017 , 0.262866 ,0); 
 \coordinate (A4) at (-0.5 , -0.688191 ,0);
 \coordinate (A5) at (0.5 , -0.688191 ,0);

\coordinate (B1) at (0,0,1);
\coordinate (B2) at (0,0,-1);

\draw (A1) -- (A2) -- (A3) -- (A4) -- (A5) -- cycle;
\draw (B1) -- (A1) -- (B2);
\draw (B1) -- (A2) -- (B2);
\draw (B1) -- (A3) -- (B2);
\draw (B1) -- (A4) -- (B2);
\draw (B1) -- (A5) -- (B2);
\draw (B1) -- (B2);

\foreach \v in {A1,A2,A3,A4,A5,B1,B2}  \filldraw[fill=white,draw=black] (\v) circle (1pt) ;
\node[right=.5mm] at (A1) {$1$};
\node[above left] at (A2) {$2$};
\node[left=.5mm] at (A3) {$3$};
\node[below left] at (A4) {$4$};
\node[right=.5mm] at (A5) {$5$};

\node[above=1mm] at (B1) {$6$};
\node[below=1mm] at (B2) {$7$};
\end{tikzpicture}
\caption{Pentagonal bipyramid}\label{fig:5}
\end{figure}

Observe that the chromatic number of $G$
is five, so $G$ does admit a quantum five coloring.
We claim that $G$ does not admit a quantum four coloring.
For this we apply our Nullstellensatz \Cref{thm:AB} to the 
quantum 4-coloring Equations \eqref{eq:graph20} and \eqref{eq:graph2}.

Generate the ideal 
\beq
\ide_1=(e^i_a e^j_a \mid (i,j)\in E,\ a\in\{1,\ldots,4\})
\subseteq\uGA(1).
\eeq
By \Cref{cor:2sidedIarchimedean}, the graph $G$ admits a quantum $4$-coloring
iff
\beq\label{eq:graph4}
-1\not\in\SOS_{\uGA(1)}+\ide_1.
\eeq

We show there are quadratic elements $s_j$ in $\uGA(1)$
so that
\beq\label{eq:graph5}
1+\sum s_j^*s_j \in \ide_1.
\eeq
To search for these elements $s_j$ we employ NC Gr\"obner basis combined with semidefinite programming.

Firstly, we lift our problem into the free algebra
$\faee$, where $e$ denotes the tuple $e=(e^i_a)$.
For this let $\Pi:\faee\to\uGA(1)$ denote the canonical epimorphism,
and let $\ide=\Pi^{-1}(\ide_1)$.
Next one computes a GB for $\ide$; with respect to a lex order
it has 350 elements of degree $\leq3$. 

To search for nc polynomials $s_j$ of degree $d$ so that
\beq\label{eq:graph6}
1+\sum s_j^*s_j \in \ide,
\eeq 
one employs the Gram matrix method
and semidefinite programming \cite{klepBook}. That is,
letting $W_d$ be the set of all monomials of degree $\leq d$
in $\faee$ (listed w.r.t. some ordering),
\Cref{eq:graph6} is equivalent to the existence of a positive
semidefinite matrix $M$ so that
\beq\label{eq:graph7}
1+ W_d^* M W_d \in \ide.
\eeq
Since ideal membership can be described using linear equations in terms of the entries of $M$ (given a Gr\"obner basis), \Cref{eq:graph7} immediately transforms into a semidefinite program (SDP).

In our example, \Cref{eq:graph7} is infeasible for $d=1$ and does
have a solution for $d=2$. Reducing $W_d$ modulo the GB (which one can do without loss of generality to help reduce the size of the SDP)
yields an SDP of size $272\times 272$.
Observing that $W_d^* M W_d=\tr(M W_dW_d^*)$ and reducing
entries of $W_dW_d^*$ modulo the GB, 
\Cref{eq:graph7}
converts into a set of linear equations on the entries of $M$
which are sparse and can thus be efficiently solved.
We are then left with
an SDP of manageable size. 
Solving the SDP 
using a standard solver
with trivial objective function yielded a 
floating point
positive definite solution $M$ with minimal eigenvalue of $\approx 10^{-2}$. By choosing a fine enough rationalization 
\cite{pabloQ,klepQ}
we thus obtain a symbolic (i.e., in exact arithmetic) positive definite solution $M$ of
\Cref{eq:graph7},  establishing that $G$ does not admit a quantum $4$-coloring.

\bibliographystyle{alpha}
\bibliography{ref}

\newpage 

\centerline{ NOT FOR PUBLICATION}

\tableofcontents

\printindex


\newpage
{\Large \bf Table of Notation}

\begin{table}[h]
    \centering
    \begin{tabular}{r|l}
$\fa$ & free algebra on $x$ + variants. \\  
$\ide$ & capital gothic letters for two-sided ideals \quad AND \\
$\ide(generators)_{algebra}$ & to be explicit in terms of gens/algebra\\
$\lide$, $\cRI$ & capital gothic letters for left/right ideals \quad AND \\
$\lide(generators)_{algebra}$ & to be explicit in terms of gens/algebra\\
$\projrset(\alpha)$ & The set of projectors used by player $\alpha$ in a nonlocal game. \\  
$\projr{i}{a}{\alpha}$ & The projector in a nonlocal game strategy corresponding to \\
&player $\alpha$ giving a response $a$ to question $i$. \\
$e(\alpha)^i_{a}$ & Formal variables satisfying the same relations as $\projr{i}{a}{\alpha}$ \\
$X(\alpha)^i_{a}$ & The signature matrix $2\projr{i}{a}{\alpha} - 1$. \\
$x(\alpha)^i_{a}$ & The formal variable in $\uGA$ corresponding to $X(\alpha)^i_{a}$. \\
$\uGA$ & The universal game algebra formed by the $e(\alpha)^i_{a}$ \\
$\uGI$ & universal game ideal\\
$\mathcal A$, etc. & $C^*$-algebras \\
$\cexp$ & conditional expectation
\\
$\Alg_{\mathcal A}(blah)$ & subalgebra of $\mathcal A$ generated by $blah$\\
$\Alg^*_{\mathcal A}(blah)$ & $*$-subalgebra of $\mathcal A$ generated by $blah$\\
$\mathcal A^{-1}$ & invertible elements in an algebra $\mathcal A$\\

$\gr(\vec{i})$ & Set of answers corresponding to ``valid responses" to question $\vec{i}$.\\
$\br(\vec{i})$ & Set of answers corresponding to ``invalid responses" to question $\vec{i}$ \\
$\validpoly$ & Set of elements $e(\alpha)^i_a$ corresponding to ``valid responses - 1". \\
$\invalidpoly$ & Set of elements $e(\alpha)^i_a$ corresponding to ``invalid responses". \\
   \end{tabular}
    \caption{Table of Notation}\label{tab:NOTATION}
\end{table}

\end{document}

%% file: packages.tex
\usepackage{amsmath}
\usepackage{amssymb}
\usepackage{amsthm}
\usepackage{amsfonts}
\usepackage[pagebackref]{hyperref}
\hypersetup{
    colorlinks=true, 
    linkcolor=blue, 
    citecolor=blue, 
    urlcolor=blue 
}
\usepackage{microtype}

\usepackage{accents}
\usepackage[boxed,titlenumbered]{algorithm2e}
\usepackage[toc,page]{appendix}
\usepackage{array}
\usepackage[utf8]{inputenc}
\usepackage[main=english, russian]{babel}
\usepackage{braket}
\usepackage{booktabs}
\usepackage{bigdelim}
\usepackage{cancel}
\usepackage[shortlabels]{enumitem}
\usepackage{float}
\usepackage[margin=1in]{geometry}
\usepackage{graphicx}
\usepackage[capitalize,nameinlink]{cleveref} 
\usepackage{mathtools}
\usepackage{multirow}
\usepackage{physics}
\usepackage{tabularx}
\usepackage[dvipsnames]{xcolor}
\usepackage{tikz}
\usepackage[textsize=footnotesize]{todonotes}
\usepackage{qcircuit}
\usetikzlibrary{positioning}
\usetikzlibrary{decorations.pathreplacing}
\usetikzlibrary{cd}
\usetikzlibrary{arrows.meta}
\usetikzlibrary{calc}
\usepackage{bbold}

\usepackage{tikzsymbols}

\usepackage{xspace}

\renewcommand{\backref}[1]{}

\renewcommand{\backrefalt}[4]{%
\ifcase #1 %
\or
[p.\ #2]%
\else
[pp.\ #2]%
\fi}

\makeatletter
\renewcommand{\paragraph}{%
 \@startsection{paragraph}{4}%
 {\z@}{2.25ex \@plus .5ex \@minus .3ex}{-1em}%
 {\normalfont\normalsize\bfseries}%
}
\makeatother

\makeatletter
\newcommand{\para}{%
 \@startsection{paragraph}{4}%
 {\z@}{2ex \@plus 3.3ex \@minus .2ex}{-1em}%
 {\normalfont\normalsize\bfseries}%
}
\makeatother

%% file: adam.tex
\newtheorem{lem}{Lemma}[section]
\newtheorem*{lem*}{Lemma}
\newtheorem{thm}[lem]{Theorem}
\crefname{thm}{theorem}{theorems}
\crefname{lem}{lemma}{lemmas}
\newtheorem*{thm*}{Theorem}
\crefname{thm*}{theorem}{theorems}
\newtheorem{prop}[lem]{Proposition}
\newtheorem{claim}[lem]{Claim}
\newtheorem{cor}[lem]{Corollary}
\newtheorem{defn}[lem]{Definition}

\theoremstyle{definition}

\newtheorem*{rmk}{Remark}

\theoremstyle{remark}

\newtheorem*{note*}{Note}
\newtheorem{ex}[lem]{Example}

\makeatletter
\newtheorem*{rep@theorem}{\rep@title}
\newcommand{\newreptheorem}[2]{%
\newenvironment{rep#1}[1]{%
 \def\rep@title{#2 \ref{##1}}%
 \begin{rep@theorem}}%
 {\end{rep@theorem}}}
\makeatother

\newreptheorem{thm}{Theorem}
\newreptheorem{lem}{Lemma}

\def\beq{\begin{equation}}
\def\eeq{\end{equation}}

\def\ben{\begin{enumerate}}
\def\een{\end{enumerate}}

\def\bem{\begin{bmatrix}}
\def\eem{\end{bmatrix}}

\def\cA{{\mathcal A}}
\def\cH{{\mathcal H}}
\def\cF{{\mathcal F}}

 \numberwithin{equation}{section}